\documentclass[a4paper,11pt]{article}

\usepackage[T1]{fontenc}
\usepackage[utf8]{inputenc}
\usepackage{fullpage}
\usepackage{times}
\usepackage{microtype}
\usepackage{amsmath,amssymb,amsthm}
\usepackage{booktabs}
\usepackage{float}
\usepackage{algorithm}
\usepackage[noend]{algpseudocode}
\usepackage{enumitem}
\usepackage[small]{caption}
\usepackage{authblk}
\usepackage{natbib}
\usepackage[hidelinks]{hyperref}
\usepackage{cleveref}
\usepackage{xcolor}

\newtheorem{theorem}{Theorem}[section]

\newtheorem{lemma}[theorem]{Lemma}
\newtheorem{corollary}[theorem]{Corollary}
\theoremstyle{definition}
\newtheorem{definition}[theorem]{Definition}

\theoremstyle{remark}

\newcommand{\ceil}[1]{\left\lceil #1\right\rceil}
\newcommand{\floor}[1]{\left\lfloor #1\right\rfloor}
\newcommand{\pos}[1]{\left(#1\right)_{+}}
\newcommand{\MMS}{\operatorname{MMS}}
\allowdisplaybreaks

\title{\bf Maximin Share Allocations Beyond Weakly \\ Lexicographic Valuations}

\author[]{Nicholas Teh}
\affil[]{University of Oxford, UK}

\date{\vspace{-10mm}}

\begin{document}
\maketitle

\begin{abstract}
We study the fair division of indivisible goods among agents with nonnegative additive valuations. An agent's maximin share (MMS) is the value she can guarantee by partitioning the goods into as many bundles as there are agents and receiving a least-valued bundle. We prove the existence of exact MMS allocations for a valuation domain that generalizes several few-valued domains as well as weakly lexicographic valuations. For each agent, the lower-valued goods, after a positive rescaling, may take values in $\{0,1,d,d+1\}$, $\{0,1,2,2e\}$, or $\{0,1,2,3,4\}$, where $d,e\ge 2$ are integers. Above these goods, the valuation may contain arbitrarily many value classes, subject to the condition that a single good in each class is worth at least the total value of all strictly lower-valued goods. The value classes, scaling factors, and parameters may vary across agents. We give an algorithm that computes an MMS allocation in $O(nm\log^2(2+m))$ time, where $n$ is the number of agents and $m$ is the number of goods.
\end{abstract}

\section{Introduction}

Fair division studies how resources can be allocated among agents with different preferences. When the resources are indivisible, even an equal division of their number need not be fair: two agents may attach very different values to the same good. For additive valuations, the \emph{maximin share} (MMS), introduced by \citet{budish2011}, provides an agent-specific comparison. An agent partitions the goods into as many bundles as there are agents and receives a bundle of minimum value according to her valuation. Her MMS is the greatest value that she can guarantee in this way. An allocation is MMS-fair if every agent receives at least her own MMS.

However, MMS allocations need not exist for general additive valuations \citep{procacciawang2014}. This has led to two main approaches: approximating each agent's MMS, and identifying valuation domains on which exact MMS can be attained. Approximation algorithms apply broadly; for example, \citet{huangzhou2025} obtain a $7/9$ guarantee and a polynomial-time $(7/9-\varepsilon)$ guarantee for every $\varepsilon>0$. Exact results, on the other hand, can reveal which features of preferences make the full guarantee possible. We pursue the second approach.

Two well-studied restrictions offer different reasons to expect exact MMS fairness. Under \emph{few-valued valuations}, each agent uses only a small number of values for individual goods. Exact MMS is known for the values $\{0,1,2\}$ \citep{amanatidis2017} and for several forms of bivalued valuations \citep{feige2022,ebadian2022,zhang2026}. Under \emph{weakly lexicographic valuations}, goods of equal value form classes, and an agent values one good in a higher class more than all goods in lower classes together. There may be arbitrarily many classes, but lower-valued goods cannot compensate for the loss of a higher-valued good. Exact MMS can also be computed in this domain \citep{ebadian2022}.

These restrictions leave a natural possibility between them. An agent may have several high-priority classes but make ordinary additive tradeoffs among her lower-valued goods. For example, she may value any computer more than all available office supplies together, while valuing a notebook at four and each pen at two. Moreover, another agent may place the same computer among her lower-valued goods. Can exact MMS be guaranteed when the lexicographic comparisons apply only above such a lower part?

\subsection{Our Results}

We answer this question affirmatively for three families of lower valuations. After a positive scaling, the values in the lower part may belong to any one of
\begin{equation}\label{eq:lower-values}
 \{0,1,d,d+1\},\quad
 \{0,1,2,2e\},\quad
 \{0,1,2,3,4\},
 \quad d,e\in\mathbb Z_{\ge2}.
\end{equation}
There is no bound on the number of goods at any of these values. An agent may then have any finite number of higher value classes. At every such class, one good must be worth at least the total value of all strictly lower-valued goods. We call this construction a \emph{lexicographic extension}, allowing equality in the comparison with all lower-valued goods. The formal definition appears in Section~\ref{sec:preliminaries}.

Our main theorem gives an exact MMS allocation in $O(nm\log^2(2+m))$ time for every instance in which each agent has such a valuation. All choices in the description are agent-specific: the lower family, the scale, the integer parameter, the number and values of the higher classes, and which goods belong to them. The theorem therefore applies even when the agents disagree on every class boundary. The algorithm runs in polynomial time. The lower families give positive cases of the question of \citet[Open Problem~6]{amanatidis2023survey} about further structured valuation domains with more than two values that admit exact MMS allocations.

The result contains the MMS existence guarantees for $\{0,1,2\}$ valuations, personalized bivalued valuations with integer ratios, and weakly lexicographic valuations. 
It also allows an agent to use all four positive values $1,2,3,4$ in her lower part, where the weakly lexicographic condition need not hold. For example, start with two goods at each of these four values. Their total is $20$. Add two goods worth $20$, then two worth $60$, then two worth $180$, and continue in the same way. Every added class meets our condition with equality, and the number of distinct values can grow without bound. Yet the lower part allows a good worth four to be replaced by goods worth three and two. The higher classes satisfy lexicographic comparisons, while the lower goods retain additive tradeoffs.

The algorithm allocates bundles one at a time. It first uses the standard reduction in which every agent's singleton values are sorted separately, so all agents agree on the order of the resulting rank goods. With $k$ agents remaining, it finds the smallest $q$ for which the ranks $1,k+1,\dots,(q-1)k+1$ meet some agent's target. Among acceptable bundles whose $j$-th rank is at least $(j-1)k+1$, it then moves the selected ranks as far right as possible, prioritizing the first rank, then the second, and so on. Any agent who accepts the resulting bundle may receive it.

The main challenge is ensuring that a bundle assigned to one agent leaves suitable partitions for every other agent. Our proof connects two facts about the lower valuations. First, a bundle chosen far enough to the right can be removed while retaining one fewer target bundle. Second, removing the first acceptable spaced bundle together with one additional high-ranked good still leaves two fewer target bundles. The latter fact controls what happens when an allocation removes fewer goods from an upper class and correspondingly more goods from lower classes. Every additional upper good left over can replace one lower target bundle. Section~\ref{sec:extension} develops this argument and shows why it can be repeated through arbitrarily many classes.

\subsection{Related Work}

The literature on MMS and other fairness notions for indivisible goods is surveyed by \citet{amanatidis2023survey}. We first discuss the valuation domains most closely related to our result. For bivalued valuations, \citet{feige2022} proves exact MMS existence when all agents use a common pair of values. \citet{ebadian2022} allow agent-specific pairs whose larger value is an integer multiple of the smaller one. More recently, \citet{zhang2026} gives a polynomial-time algorithm for arbitrary personalized pairs. Our main result includes the integer-ratio domain but not all arbitrary personalized pairs. Its additional scope comes from allowing three or four positive lower values together with unboundedly many higher classes.\footnote{Restricted valuation domains generally permit stronger guarantees for allocation rules; see, e.g., some works on binary valuations \citep{brandl2026binarycharacterizations,HalpernPrPs20,neohteh2025efx,suksompongteh2022weightednash,suksompongteh2023matroid,montanari2025weightedenvy} or identical valuations \citep{choo2024house,elkind2024chorebudgets,plautroughgarden2020,scarlett2023oneandall}.}

Lexicographic preferences have also been studied in conjunction with envy-freeness and efficiency. \citet{hosseini2021lexicographic} consider strict lexicographic preferences. For weakly lexicographic preferences, \citet{ebadian2022} compute MMS allocations that are also Pareto optimal, meaning that no allocation makes one agent better off without making another worse off. \citet{aziz2023possible} and \citet{hosseini2024weakly} further obtain MMS together with Pareto optimality and envy-freeness up to any good (EFX). The last property requires that an agent not envy another agent after any single good is removed from the latter's bundle. Polynomial-time EFX algorithms are also known for personalized bivalued goods \citep{byrka2026,jintao2026}. Our result extends the valuation domain for exact MMS; it does not claim these additional guarantees for the proposed algorithm.

A different restriction is that of \emph{factored valuations}: after scaling, each positive singleton value divides the next larger positive value. The family $\{0,1,2,2e\}$ in \eqref{eq:lower-values} is factored. \citet{ebadian2022} show how to compute an individual's MMS for all factored valuations, but their exact allocation results concern the bivalued and weakly lexicographic subclasses. The exact MMS theorem for all factored instances of \citet{garg2025chores} concerns chores, not goods. Our domain need not be factored: the values $2$ and $3$ can both occur in the lower part, and no divisibility condition is imposed on higher values.

Our algorithm also builds on techniques from the broader MMS literature. The ordered-instance reduction is due to \citet{bouveretlemaitre2016}; we use the target-preserving formulation of \citet{barmankrishnamurthy2020}. Algorithms that fill bundles until an agent reaches a target are common in MMS approximation \citep{ghodsi2018,gargetal2019}, as are valid reductions that preserve the remaining agents' MMS values \citep{ghodsi2018,gargtaki2021}.

More directly, \citet{ebadian2022} use the spaced bundles $\{1,k+1,2k+1,\dots\}$ in their exact MMS algorithms for weakly lexicographic and personalized integer-ratio bivalued valuations. Their proofs place such bundles within MMS partitions. \citet{feige2022} chooses a smallest bundle in an MMS partition and takes the latest goods within each of its two value classes. In approximation algorithms, \citet{heidari2026} move the final selected rank as far right as acceptability permits, and \citet{huangzhou2025} also vary an earlier rank when selecting triples. Our rule chooses the whole rank vector lexicographically for arbitrary bundle sizes. The new allocation argument establishes the two partition properties in Section~\ref{sec:extension} for the three lower families and proves that they are preserved when higher classes are added.

Lexicographic comparisons with bundles feasible for a single agent also appear in the analysis of chore-allocation algorithms \citep{huangsegal2023,garg2025chores}. Our comparison is between increasing rank vectors of a fixed size under the bounds in \eqref{eq:rank-family}, and is used to preserve lower bounds on the remaining bundles' values.

For an individual valuation with a bounded number of distinct singleton values, count-based dynamic programming already yields polynomial-time MMS computation \citep{nguyenrothe2020,nguyenrothe2023}. The MMS calculation across value classes for weakly lexicographic preferences is also known \citep{aziz2023possible,hosseini2024weakly}. We give explicit formulas for the lower families and allow an arbitrary lower valuation in the calculation for one added class. These computations accompany the simultaneous exact allocation guarantee, including the interaction between a few-valued lower part and arbitrarily many higher classes.

\section{Preliminaries}\label{sec:preliminaries}

An instance is $I=(N,G,(v_i)_{i\in N})$, where $N=\{1,\dots,n\}$ is the set of agents and $G$ is a set of $m$ indivisible goods. A \emph{bundle} is a subset of $G$. Write $v_i(g)$ for $v_i(\{g\})$. Each valuation is nonnegative and additive: $v_i(S)=\sum_{g\in S}v_i(g)$ for every $S\subseteq G$. An allocation $A=(A_i)_{i\in N}$ is a partition of $G$; empty bundles are allowed.
For a valuation $v$ on a set $U$ and an integer $k\ge1$, let $\MMS_k(v,U)=\max_{(P_1,\dots,P_k)\text{ a partition of }U} \min_{1\le j\le k}v(P_j)$.
Agent $i$'s MMS in the original instance is $\mu_i=\MMS_n(v_i,G)$. An allocation is \emph{MMS-fair} if $v_i(A_i)\ge\mu_i$ for every agent $i$. We also use a \emph{target partition}: a partition into a specified number of bundles, each worth at least a specified target to the valuation under consideration. Thus a $k$-partition at target $t$ exists exactly when $t\le\MMS_k(v,U)$.

Positive scaling of a valuation scales its MMS and all bundle values by the same amount. We may therefore normalize a lower valuation by its scale when proving statements about \eqref{eq:lower-values}. Bundle values in a normalized lower valuation are integers, so replacing a real target $t$ by $\ceil{t}$ changes neither acceptability nor the existence of a target partition. Values in higher classes, however, need not be integer multiples of the lower scale.

\paragraph{The valuation domain.}
A \emph{value class} of $v$ consists of all goods having a specified positive value. The following definition makes precise the extension used throughout the paper.

\begin{definition}[Lexicographic extension]\label{def:domain}
A valuation $v$ is a lexicographic extension of its restriction to a lower part $C\subseteq G$ if the goods outside $C$ form nonempty value classes $L_1,\dots,L_s$, with values $\max_{g\in C}v(g)<H_1<\cdots<H_s$ and $v(g)=H_j$ for all $g \in L_j$,
and $H_j \ge \sum_{\{g\in G:\ v(g)<H_j\}}v(g)$ for all $j = 1,\dots, s$.
The list of higher classes may be empty. All zero-valued goods belong to $C$, and the maximum over an empty lower part is interpreted as zero.
\end{definition}

The sum in Definition~\ref{def:domain} includes every lower-valued good, with its multiplicity. Thus a value class may contain any number of goods, but those goods contribute to the requirement on each higher class. When all inequalities are strict and $C$ has one positive value, the definition gives the usual weakly lexicographic domain \citep[Definition~2.2]{ebadian2022}, with zero-valued goods also allowed. We permit equality as well. This use of equality is distinct from the word ``weakly'' in weakly lexicographic preferences, which refers to ties among goods within a class.

We study instances in which, for each agent $i$, the lower part $C_i$ has singleton values in a positive scaling of one of the sets in \eqref{eq:lower-values}. The definition does not require the agents to agree on their lower parts or on any value class. Deleting goods preserves the condition: the lower values remain in the same set, and each sum in Definition~\ref{def:domain} can only decrease.

\paragraph{Fixed targets and ordered instances.}
The algorithm accepts fixed targets $0\le t_i\le\mu_i$. A bundle is \emph{acceptable to agent $i$} if its value to her is at least $t_i$, and \emph{acceptable} if it is acceptable to at least one remaining agent. Agents with target zero are assigned empty bundles and removed before the first iteration. For an exact MMS allocation we set $t_i=\mu_i$; these targets are not recomputed during the algorithm.

Sort each agent's singleton values separately in nonincreasing order. The auxiliary goods represent ranks, so all agents agree on their order, although their values at a rank may differ. This preserves each agent's multiset of singleton values, her MMS, and membership in our domain \citep{bouveretlemaitre2016,barmankrishnamurthy2020}.

An allocation of auxiliary ranks can be transferred back without decreasing any agent's value \citep[Theorem~3.2]{barmankrishnamurthy2020}. Process ranks in increasing order. When rank $r$ belongs to agent $i$, let her choose a most valuable remaining original good. At most $r-1$ original goods have already been taken, so her choice is worth at least her $r$-th largest singleton value. We therefore work with ordered instances below. Ties may be broken arbitrarily.

Let $U$ denote the current set of unallocated auxiliary goods, let $k$ be the number of remaining agents, and write $u_j$ for the $j$-th good in their common order. A \emph{prefix} of length $X$ is the set $\{u_1,\dots,u_X\}$. When used in a partition argument, a rank beyond $|U|$ denotes a distinct zero-valued dummy good. These dummy goods are never required by the allocation algorithm.

\section{The Algorithm}\label{sec:algorithm}
We now introduce our algorithm. Intuitively, it seeks a bundle that meets an agent's target while leaving enough of the more valuable goods for those who remain. With $k\ge2$ agents, consider the \emph{spaced bundles} used by \citet[Sections~4.3--4.4]{ebadian2022}: $S_q=\{u_1,u_{k+1},\dots,u_{(q-1)k+1}\}$ and $S_0=\varnothing$.
Let $q$ be the least index for which $S_q$ is acceptable. Among bundles of this size, consider
\begin{equation}\label{eq:rank-family}
 \mathcal F_{k,q}=
 \{\{u_{r_1},\dots,u_{r_q}\}:\
1\le r_1<\cdots<r_q\le |U|, \ r_j\ge(j-1)k+1\text{ for every }j\}.
\end{equation}
The bounds ensure that a selected bundle uses at most $\ceil{X/k}$ goods from any prefix of length $X$.

We choose the acceptable member of $\mathcal F_{k,q}$ with the \emph{lexicographically largest} increasing rank vector: maximize $r_1$, subject to acceptability; then maximize $r_2$ with $r_1$ fixed; and continue. Equivalently, at the first coordinate where two vectors differ, prefer the larger rank. We call this the \emph{rightmost acceptable bundle}. The comparison is between rank vectors, not between the agents' preferences over bundles.

\begin{algorithm}[!ht]
\caption{Exact MMS allocation by rightmost acceptable bundles}\label{alg:exact}
\begin{algorithmic}[1]
\Require An ordered instance and targets $0\le t_i\le\mu_i$.
\State Assign an empty bundle to each agent with $t_i=0$ and remove her.
\While{at least two agents remain}
    \State Let $k$ be the number of remaining agents and $U$ the unallocated goods.
    \State Find the least $q$ for which $S_q$ is acceptable.
    \State Choose the acceptable $B\in\mathcal F_{k,q}$ with largest rank vector.
    \State Give $B$ to any agent $i$ with $v_i(B)\ge t_i$; remove $i$ and $B$.
\EndWhile
\State Give all remaining goods to the last agent, if there is one.
\State If all targets were zero, assign all goods arbitrarily.
\State Transfer the allocation of ranks to the original goods.
\end{algorithmic}
\end{algorithm}

The search in each iteration is well defined whenever every remaining agent has a target partition. In each consecutive group of at most $k$ goods, the first good is worth at least the group's average. Thus, the full spaced sequence is worth at least $v_i(U)/k\ge t_i$ to every remaining agent $i$, and $q\le\ceil{|U|/k}$ exists. Furthermore, an acceptable selected bundle contains no good worth zero to its recipient. Deleting such a good would give an acceptable $(q-1)$-bundle whose $j$-th rank is at least $(j-1)k+1$, contradicting the choice of $q$.

\begin{theorem}\label{thm:main}
Suppose that each agent's valuation is a lexicographic extension of a lower valuation whose singleton values, after a positive scaling, lie in one of the sets in \eqref{eq:lower-values}. Then an exact MMS allocation exists. Algorithm~\ref{alg:exact} finds one in $O(nm\log^2(2+m))$ time and $O(nm)$ space, including the computation of the MMS values.
The algorithm also meets any fixed targets $0\le t_i\le\mu_i$, with any choice of accepting recipient.
\end{theorem}

In fact, each removal preserves or increases the remaining agents' MMS values in the ordered instance. If $B$ is selected with $k\ge2$ agents and remaining goods $U$, then every agent $i$ who remains after the allocation satisfies
\begin{equation}\label{eq:monotonicity}
 \MMS_{k-1}(v_i,U\setminus B)\ \ge\ \MMS_k(v_i,U).
\end{equation}
The final transfer guarantees each agent's target in the original instance.

\paragraph{Changing only the final rank is insufficient.}
Moving the final selected rank farther right is used by \citet{heidari2026}, and the triple selection of \citet{huangzhou2025} also varies an earlier rank while keeping the first good fixed. Moving only the final selected rank need not suffice in our domain. Consider two identical agents and goods worth $13,13,13,13,3,3,2,2,2$.
The lower goods total $12$, so the class worth $13$ meets the condition in Definition~\ref{def:domain}. Both agents have MMS $32$, attained by $(13+13+3+3,\;13+13+2+2+2)$. The first acceptable spaced size is five. A bundle with values $13,13,3,2,2$ is acceptable but leaves only $31$; replacing just its last good cannot reduce its value. Algorithm~\ref{alg:exact} instead selects ranks $(3,4,7,8,9)$, worth $32$, and leaves $32$. It is the lexicographic choice of all selected ranks, rather than a change to the final selected good alone, that gives the required guarantee.

\section{Preserving MMS Across Value Classes}\label{sec:extension}

We now explain why the algorithm remains valid when higher value classes are added. The proof works with one agent's valuation at a time. Fix an ordered valuation $v$, a number $k\ge2$, and a target $t>0$ that admits a $k$-partition. Let $q_v$ be the least $q$ with $v(S_q)\ge t$.

Two partition properties of the lower valuations are sufficient. We state them together to make their different roles clear.
\begin{enumerate}[label=(\roman*),leftmargin=2.2em]
\item\label{item:one} For every $1\le q\le q_v$, removing any $B\in\mathcal F_{k,q}$ whose rank vector is at least that of every $t$-acceptable member of $\mathcal F_{k,q}$ leaves a $(k-1)$-partition at $t$. When no member meets $t$, the comparison imposes no further condition on $B$.
\item\label{item:two} For $k\ge3$, removing $S_{q_v}\cup\{u_2\}$ leaves a $(k-2)$-partition at $t$.
\end{enumerate}
Property~\ref{item:one} gives the partition preservation required for a valid reduction \citep{ghodsi2018,ebadian2022}. It is what the allocation algorithm needs, and applies even if the selected bundle is unacceptable to the agent whose remaining partition is being considered. Property~\ref{item:two} allows one additional high-ranked good to be removed at the cost of one additional target bundle. Its purpose is to account for a shortage of selected upper-class goods.

\begin{theorem}\label{thm:extension}
Suppose a family of valuations is closed under deleting goods and satisfies properties~\ref{item:one} and~\ref{item:two} for every feasible target and every number of bundles. Add any number of goods of a common value $H$, larger than every existing singleton value, with $H$ at least the total value of the existing goods. The resulting family also satisfies both properties. Consequently, either property continues to hold after any finite sequence of such additions.
\end{theorem}

Theorem~\ref{thm:extension} is the reason that the number of value classes need not be bounded. Its assumptions are substantive. For values $(5,5,3,3,1,1)$, the target $6$ is feasible for three bundles, using $(5+1,5+1,3+3)$. However, $S_{q_v}\cup\{u_2\}$ consists of the two fives and one three, leaving value five. Thus property~\ref{item:two} does not follow from the existence of a target partition for an arbitrary additive valuation.

\subsection{MMS after Adding a Higher Value Class}

The MMS computation across value classes for weakly lexicographic preferences appears in \citet[Lemma~1]{aziz2023possible} and \citet[Proposition~1]{hosseini2024weakly}. It distributes a class evenly and continues with the bundles receiving fewer goods from that class. The following identity permits an arbitrary lower valuation and equality in $H\ge v(U)$. It also explains which lower target the preservation proof must retain.

\begin{lemma}\label{lem:recursion}
Let $L$ contain $a$ goods of value $H$, and let $U$ consist of lower-valued goods, with $H\ge v(U)$. Write $a=k\ell+r$, where $0\le r<k$. Then $\MMS_k(v,U\cup L)=\ell H+\MMS_{k-r}(v,U)$.
\end{lemma}
\begin{proof}
Give $\ell$ goods from $L$ to each bundle and one additional such good to $r$ bundles. Partition $U$ among the other $k-r$ bundles at its $(k-r)$-MMS. Since $H\ge v(U)\ge\MMS_{k-r}(v,U)$, every bundle reaches the stated lower bound.

For the reverse inequality, consider a partition with minimum value $\tau>\ell H$. Each bundle must contain at least $\ell$ goods from $L$: a bundle with at most $\ell-1$ has value at most $(\ell-1)H+v(U)\le\ell H$. At least $k-r$ bundles contain exactly $\ell$ goods from $L$, so their disjoint lower parts all have value at least $\tau-\ell H$. Completing these parts to a partition of $U$ gives $\tau-\ell H\le\MMS_{k-r}(v,U)$. Partitions with minimum value at most $\ell H$ already satisfy the same upper bound.
\end{proof}

Lemma~\ref{lem:recursion} reduces exact MMS computation to the lower part. Applying it repeatedly uses only one step per upper class. For allocation, however, the lower goods that remain depend on the bundle selected for another agent. We must account for those removals as well.

\subsection{Proof Idea for the Extension Theorem}

The main case in proving property~\ref{item:one} in Theorem~\ref{thm:extension} is when the number of upper goods is divisible by $k$. Suppose there are exactly $k\ell$ goods worth $H$, and write the target as $t=\ell H+s$, where $s>0$. Lemma~\ref{lem:recursion} ensures that the lower goods have a $k$-partition at $s$. Let $q_U$ be their first acceptable spaced size, so the full first acceptable size is $\ell+q_U$.

A selected bundle satisfying the rank bounds removes at most $\ell$ upper goods. Suppose it removes $\ell-e$ of them. Its lower part then contains at most $q_U+e$ goods and at most $\ceil{X/k}+e$ goods from any lower prefix of length $X$. Lemma~\ref{lem:additional-removals} in Appendix~\ref{app:extension} shows that property~\ref{item:two} allows these lower removals at the cost of at most $e+1$ target bundles. Thus at least $k-e-1$ lower target bundles remain when $1\le e\le k-2$.

After giving $\ell$ upper goods to each of the $k-1$ remaining bundles, the number of upper goods left over is exactly $k\ell-(\ell-e)-(k-1)\ell=e$.
Each is worth at least $s$, so these $e$ goods replace the lower target bundles that may have been lost. The other $k-e-1$ bundles use the lower partition at $s$. All $k-1$ bundles therefore reach $t$. If $e\ge k-1$, the upper goods suffice on their own.

When $e=0$, there are no extra upper goods. This is where the rightmost choice is needed. Keep the selected upper goods fixed. Any later lower bundle meeting $s$ would give a later full bundle meeting $t$, contrary to the choice of the selected bundle. Property~\ref{item:one} for the lower valuation therefore applies directly. The other cases, in which the upper count is not divisible by $k$ or the upper goods alone meet the target, require only distributing upper goods and discarding affected lower partition bundles. The complete proof, including preservation of property~\ref{item:two}, is in Appendix~\ref{app:extension}.

\subsection{The Two Properties for the Lower Valuations}
It remains to establish the two properties for the lower valuations. These require a different argument from the upper classes: one good need not outweigh all smaller goods. We derive exact conditions for a target partition, and show that the removals in the two properties preserve these conditions.

\begin{lemma}\label{lem:lower-properties}
Every positive scaling of a valuation with singleton values in one of the sets in \eqref{eq:lower-values} satisfies properties~\ref{item:one} and~\ref{item:two}.
\end{lemma}

Appendix~\ref{app:partitions} establishes the partition conditions. They measure not only total value but also the availability of goods needed to complete bundles. For $\{0,1,d,d+1\}$, distributing the $d$- and $(d+1)$-goods evenly minimizes the number of ones needed. For $\{0,1,2,2e\}$, the even upper value allows two corresponding deficit measures to be minimized by the same distribution. For $\{0,1,2,3,4\}$, balanced counts of nonunit goods and the useful extra value of threes and fours give the required conditions.

Appendix~\ref{app:one} uses these conditions to prove property~\ref{item:one}. The rightmost choice is important even in the lower parts. After normalizing values and rounding the target as in Section~\ref{sec:preliminaries}, if the selected value exceeds the target, replacing a selected good by a later good of value one less would preserve acceptability. The rightmost rule excludes that possibility. Together with the prefix bounds, this limits the remaining value patterns enough to retain a target partition. Appendix~\ref{app:two} proves property~\ref{item:two} by bounding the amount removed from each partition inequality, treating the small targets separately where needed.

\subsection{Completing the Proof of Theorem~\ref{thm:main}}
The two properties for the lower valuations and Theorem~\ref{thm:extension} now give the preservation required by the algorithm.

\begin{proof}[Proof of Theorem~\ref{thm:main}]
By Lemma~\ref{lem:lower-properties} and repeated application of Theorem~\ref{thm:extension}, every valuation in the theorem satisfies property~\ref{item:one}. Initially each agent has an $n$-partition at her fixed target. Removing zero-target agents without taking goods preserves the other agents' target partitions by merging parts.

Consider an iteration with $k\ge2$ agents. For each remaining agent $i$, the selected size $q$ is at most her first acceptable spaced size. Every bundle acceptable to her in $\mathcal F_{k,q}$ is also acceptable in the algorithm's search. The selected bundle is therefore at least as far right as every bundle acceptable to her in $\mathcal F_{k,q}$. Property~\ref{item:one} leaves her a $(k-1)$-partition at $t_i$. Deleting the selected goods keeps her valuation in the domain. Induction proves that every agent eventually receives a bundle meeting her fixed target, including the final agent.

For the stronger statement, fix any such agent and set $\widehat t_i=\MMS_k(v_i,U)\ge t_i$. Her first acceptable spaced size at $\widehat t_i$ is no smaller than at $t_i$, and every bundle acceptable at $\widehat t_i$ is also acceptable at $t_i$. Hence the selected bundle still meets the premises of property~\ref{item:one} at $\widehat t_i$. Applying it gives $\eqref{eq:monotonicity}$. The ordered-to-original transfer then meets all original targets. Section~\ref{sec:runtime} establishes the running-time and space bounds.
\end{proof}

\section{Efficient Implementation}\label{sec:runtime}

The algorithm need not enumerate the bundles in $\mathcal F_{k,q}$. It maximizes their selected ranks one coordinate at a time, using binary search. Testing a proposed rank through its earliest possible completion also appears in the triple selection of \citet[Algorithm~1]{huangzhou2025}; here we apply it to every coordinate and arbitrary bundle sizes. Suppose the first $j-1$ ranks are fixed, and a proposed $j$-th rank is $r$. The earliest possible completion has ranks
\begin{equation}\label{eq:completion}
 r_\ell=\max\{(\ell-1)k+1,\ r+\ell-j\},
 \quad \ell=j,\dots,q.
\end{equation}
These ranks give every agent at least as much value as any other completion of that prefix. Therefore a completion is acceptable if and only if the earliest completion is acceptable. As $r$ increases, its value to every agent weakly decreases, so binary search finds the largest admissible coordinate. Repeating the search through coordinate $q$ returns exactly the bundle specified in Algorithm~\ref{alg:exact}.

The useful feature of \eqref{eq:completion} is that it consists of a consecutive interval followed by a suffix of the spaced sequence. For each agent, maintain a binary indexed tree for sums of her values at currently present ranks. Also store the partial sums of the spaced sequence during the current iteration. A proposed completion can then be evaluated with one interval-sum query and one stored suffix sum, in $O(\log(2+m))$ operations per agent. Binary indexed trees provide interval sums and deletions in logarithmic time \citep{fenwick1994}.

Each coordinate search takes $O(k\log^2(2+m))$ operations. Every selected coordinate corresponds to a real good that is then removed, so the sum of selected bundle sizes over all iterations is at most $m$. Since $k\le n$, all coordinate searches together take $O(nm\log^2(2+m))$ operations. Maintaining the orders, forming the spaced sums, and transferring the final allocation require at most $O(nm\log(2+m))$ additional operations.

For the initial targets, first sort and group each agent's equal positive values. There are at most four distinct positive values in her lower part, so at most five possible positions for its upper end, including an empty lower part. Testing these positions and the inequalities in Definition~\ref{def:domain} takes linear time after sorting. Repeatedly apply Lemma~\ref{lem:recursion} through the upper classes, and finish with the explicit lower-part MMS formulas in Appendix~\ref{app:partitions}. This requires a constant number of operations per class; it does not search over the magnitude of a value. Total space is $O(nm)$.

Appendix~\ref{app:implementation} gives the details of recognizing the lower families and evaluating \eqref{eq:completion}, including empty and missing value classes. The stated bounds count arithmetic operations and comparisons.

\section{Discussion}\label{sec:discussion}

Our theorem allows additive tradeoffs among a few-valued lower part without requiring the higher value classes to be few in number or common across agents. The proof identifies why this combination works: an extra upper good can replace a lower target bundle, and the lower valuations bound the number of target bundles lost under additional removals. Thus the argument preserves the whole target partition, rather than only enough total value for it.

The numerical condition on upper values is important for the proposed algorithm. For two identical agents and an integer $H\ge5$, consider
\[
 H,3,3,\underbrace{1,\dots,1}_{H-2\text{ goods}}.
\]
Their MMS is $H+2$, attained by $(H+1+1,\;3+3+1^{H-4})$, where $1^{H-4}$ denotes $H-4$ unit-valued goods. The first acceptable spaced size is two. Its only acceptable value pattern is $H+3$, so Algorithm~\ref{alg:exact} leaves $H+1$ for the second agent. The lower goods lie in $\{1,3\}$ and total $H+4$. Since $H/(H+4)$ tends to one, replacing the condition in Definition~\ref{def:domain} by $H\ge(1-\varepsilon)$ times the lower total, for any fixed $\varepsilon>0$, does not retain the algorithm's guarantee. This example concerns the rule, not MMS existence: the displayed partition is itself an exact MMS allocation.

A natural next question is which other lower valuations satisfy the two partition properties in Section~\ref{sec:extension}. Any such family immediately would give an unbounded-level exact MMS domain through Theorem~\ref{thm:extension}. A separate question is whether exact MMS can be combined with envy-freeness up to any good (EFX) \citep{caragiannis2019,plautroughgarden2020} in the larger domain studied here.

\subsection*{Declaration of generative AI use}

The author used GPT-6 as a research assistant when exploring proof ideas and revising the
exposition. He checked the mathematical arguments and take full responsibility for the contents of the paper.

\bibliographystyle{plainnat}
\bibliography{bib}

\begin{thebibliography}{38}
\providecommand{\natexlab}[1]{#1}
\providecommand{\url}[1]{\texttt{#1}}
\expandafter\ifx\csname urlstyle\endcsname\relax
  \providecommand{\doi}[1]{doi: #1}\else
  \providecommand{\doi}{doi: \begingroup \urlstyle{rm}\Url}\fi

\bibitem[Amanatidis et~al.(2017)Amanatidis, Markakis, Nikzad, and Saberi]{amanatidis2017}
Georgios Amanatidis, Evangelos Markakis, Afshin Nikzad, and Amin Saberi.
\newblock Approximation algorithms for computing maximin share allocations.
\newblock \emph{ACM Transactions on Algorithms}, 13\penalty0 (4):\penalty0 52:1--52:28, 2017.

\bibitem[Amanatidis et~al.(2023)Amanatidis, Aziz, Birmpas, Filos-Ratsikas, Li, Moulin, Voudouris, and Wu]{amanatidis2023survey}
Georgios Amanatidis, Haris Aziz, Georgios Birmpas, Aris Filos-Ratsikas, Bo~Li, Herv{\'e} Moulin, Alexandros~A. Voudouris, and Xiaowei Wu.
\newblock Fair division of indivisible goods: Recent progress and open questions.
\newblock \emph{Artificial Intelligence}, 322:\penalty0 103965, 2023.

\bibitem[Atamt{\"u}rk and G{\"u}nl{\"u}k(2010)]{atamturkgunluk2010}
Alper Atamt{\"u}rk and Oktay G{\"u}nl{\"u}k.
\newblock Mingling: Mixed-integer rounding with bounds.
\newblock \emph{Mathematical Programming}, 123:\penalty0 315--338, 2010.

\bibitem[Aziz et~al.(2023)Aziz, Li, Xing, and Zhou]{aziz2023possible}
Haris Aziz, Bo~Li, Shiji Xing, and Yu~Zhou.
\newblock Possible fairness for allocating indivisible resources.
\newblock In \emph{Proceedings of the 22nd International Conference on Autonomous Agents and Multiagent Systems}, pages 197--205, 2023.

\bibitem[Barman and Krishnamurthy(2020)]{barmankrishnamurthy2020}
Siddharth Barman and Sanath~Kumar Krishnamurthy.
\newblock Approximation algorithms for maximin fair division.
\newblock \emph{ACM Transactions on Economics and Computation}, 8\penalty0 (1):\penalty0 5:1--5:28, 2020.
\newblock \doi{10.1145/3381525}.

\bibitem[Bouveret and Lema{\^\i}tre(2016)]{bouveretlemaitre2016}
Sylvain Bouveret and Michel Lema{\^\i}tre.
\newblock Characterizing conflicts in fair division of indivisible goods using a scale of criteria.
\newblock \emph{Autonomous Agents and Multi-Agent Systems}, 30\penalty0 (2):\penalty0 259--290, 2016.

\bibitem[Brandl et~al.(2026)Brandl, Suksompong, and Teh]{brandl2026binarycharacterizations}
Florian Brandl, Warut Suksompong, and Nicholas Teh.
\newblock Fair division with binary valuations: Characterizations.
\newblock In \emph{Proceedings of the 19th International Symposium on Algorithmic Game Theory (SAGT)}, 2026.
\newblock Forthcoming.

\bibitem[Budish(2011)]{budish2011}
Eric Budish.
\newblock The combinatorial assignment problem: Approximate competitive equilibrium from equal incomes.
\newblock \emph{Journal of Political Economy}, 119\penalty0 (6):\penalty0 1061--1103, 2011.

\bibitem[Byrka et~al.(2026)Byrka, Malinka, and Ponitka]{byrka2026}
Jaros{\l}aw Byrka, Franciszek Malinka, and Tomasz Ponitka.
\newblock Probing {EFX} via {PMMS}: (non-)existence results in discrete fair division.
\newblock In \emph{Proceedings of the 40th AAAI Conference on Artificial Intelligence (AAAI)}, pages 16735--16742, 2026.

\bibitem[Caragiannis et~al.(2019)Caragiannis, Kurokawa, Moulin, Procaccia, Shah, and Wang]{caragiannis2019}
Ioannis Caragiannis, David Kurokawa, Herv{\'e} Moulin, Ariel~D. Procaccia, Nisarg Shah, and Junxing Wang.
\newblock The unreasonable fairness of maximum {Nash} welfare.
\newblock \emph{ACM Transactions on Economics and Computation}, 7\penalty0 (3):\penalty0 12:1--12:32, 2019.

\bibitem[Choo et~al.(2024)Choo, Ling, Suksompong, Teh, and Zhang]{choo2024house}
Davin Choo, Yan~Hao Ling, Warut Suksompong, Nicholas Teh, and Jian Zhang.
\newblock Envy-free house allocation with minimum subsidy.
\newblock \emph{Operations Research Letters}, 54:\penalty0 107103, 2024.

\bibitem[Ebadian et~al.(2022)Ebadian, Peters, and Shah]{ebadian2022}
Soroush Ebadian, Dominik Peters, and Nisarg Shah.
\newblock How to fairly allocate easy and difficult chores.
\newblock In \emph{Proceedings of the 21st International Conference on Autonomous Agents and Multiagent Systems (AAMAS)}, pages 372--380, 2022.

\bibitem[Elkind et~al.(2024)Elkind, Igarashi, and Teh]{elkind2024chorebudgets}
Edith Elkind, Ayumi Igarashi, and Nicholas Teh.
\newblock Fair division of chores with budget constraints.
\newblock In \emph{Proceedings of the 17th International Symposium on Algorithmic Game Theory (SAGT)}, pages 55--71, 2024.

\bibitem[Federgruen and Groenevelt(1986)]{federgruengroenevelt1986}
Awi Federgruen and Henri Groenevelt.
\newblock The greedy procedure for resource allocation problems: Necessary and sufficient conditions for optimality.
\newblock \emph{Operations Research}, 34\penalty0 (6):\penalty0 909--918, 1986.

\bibitem[Feige(2022)]{feige2022}
Uriel Feige.
\newblock Maximin fair allocations with two item values, 2022.
\newblock URL \url{https://www.wisdom.weizmann.ac.il/~feige/mypapers/MMSab.pdf}.

\bibitem[Fenwick(1994)]{fenwick1994}
Peter~M. Fenwick.
\newblock A new data structure for cumulative frequency tables.
\newblock \emph{Software---Practice \& Experience}, 24\penalty0 (3):\penalty0 327--336, 1994.

\bibitem[Garg and Taki(2021)]{gargtaki2021}
Jugal Garg and Setareh Taki.
\newblock An improved approximation algorithm for maximin shares.
\newblock \emph{Artificial Intelligence}, 300:\penalty0 103547, 2021.

\bibitem[Garg et~al.(2019)Garg, McGlaughlin, and Taki]{gargetal2019}
Jugal Garg, Peter McGlaughlin, and Setareh Taki.
\newblock Approximating maximin share allocations.
\newblock In \emph{Proceedings of the 2nd Symposium on Simplicity in Algorithms (SOSA)}, pages 20:1--20:11, 2019.

\bibitem[Garg et~al.(2025)Garg, Huang, and Segal-Halevi]{garg2025chores}
Jugal Garg, Xin Huang, and Erel Segal-Halevi.
\newblock Improved maximin share approximations for chores by bin packing.
\newblock In \emph{Proceedings of the 39th AAAI Conference on Artificial Intelligence (AAAI)}, pages 13881--13888, 2025.

\bibitem[Ghodsi et~al.(2018)Ghodsi, Hajiaghayi, Seddighin, Seddighin, and Yami]{ghodsi2018}
Mohammad Ghodsi, Mohammad~Taghi Hajiaghayi, Masoud Seddighin, Saeed Seddighin, and Hadi Yami.
\newblock Fair allocation of indivisible goods: Improvements and generalizations.
\newblock In \emph{Proceedings of the 19th ACM Conference on Economics and Computation (EC)}, pages 539--556, 2018.

\bibitem[Halpern et~al.(2020)Halpern, Procaccia, Psomas, and Shah]{HalpernPrPs20}
Daniel Halpern, Ariel~D. Procaccia, Alexandros Psomas, and Nisarg Shah.
\newblock Fair division with binary valuations: One rule to rule them all.
\newblock In \emph{Proceedings of the 16th Conference on Web and Internet Economics (WINE)}, pages 370--383, 2020.
\newblock Extended version available at arXiv:2007.06073v2.

\bibitem[Heidari et~al.(2026)Heidari, Kaviani, Seddighin, and Shahrezaei]{heidari2026}
Ehsan Heidari, Alireza Kaviani, Masoud Seddighin, and AmirMohammad Shahrezaei.
\newblock Improved maximin share guarantee for additive valuations.
\newblock In \emph{Proceedings of the 2026 Annual ACM-SIAM Symposium on Discrete Algorithms (SODA)}, pages 2239--2290, 2026.

\bibitem[Hosseini et~al.(2021)Hosseini, Sikdar, Vaish, and Xia]{hosseini2021lexicographic}
Hadi Hosseini, Sujoy Sikdar, Rohit Vaish, and Lirong Xia.
\newblock Fair and efficient allocations under lexicographic preferences.
\newblock In \emph{Proceedings of the 35th AAAI Conference on Artificial Intelligence (AAAI)}, pages 5472--5480, 2021.

\bibitem[Hosseini et~al.(2024)Hosseini, Mammadov, and W{\k{a}}s]{hosseini2024weakly}
Hadi Hosseini, Aghaheybat Mammadov, and Tomasz W{\k{a}}s.
\newblock Almost envy-freeness under weakly lexicographic preferences, 2024.
\newblock arXiv:2404.19740.

\bibitem[Huang and Segal-Halevi(2023)]{huangsegal2023}
Xin Huang and Erel Segal-Halevi.
\newblock A reduction from chores allocation to job scheduling.
\newblock In \emph{Proceedings of the 24th ACM Conference on Economics and Computation (EC)}, page 908, 2023.

\bibitem[Huang and Zhou(2025)]{huangzhou2025}
Xin Huang and Shengwei Zhou.
\newblock An {FPTAS} for $7/9$-approximation to maximin share allocations, 2025.
\newblock arXiv:2511.13056.

\bibitem[Jin and Tao(2025)]{jintao2026}
Jiarong Jin and Biaoshuai Tao.
\newblock On {Pareto}-optimal and fair allocations with personalized bi-valued utilities.
\newblock In \emph{Proceedings of the 21st International Conference on Web and Internet Economics (WINE)}, pages 521--537, 2025.

\bibitem[Kurokawa et~al.(2018)Kurokawa, Procaccia, and Wang]{procacciawang2014}
David Kurokawa, Ariel~D. Procaccia, and Junxing Wang.
\newblock Fair enough: Guaranteeing approximate maximin shares.
\newblock \emph{Journal of the ACM}, 65, 2018.

\bibitem[Montanari et~al.(2025)Montanari, Schmidt-Kraepelin, Suksompong, and Teh]{montanari2025weightedenvy}
Luisa Montanari, Ulrike Schmidt-Kraepelin, Warut Suksompong, and Nicholas Teh.
\newblock Weighted envy-freeness for submodular valuations.
\newblock \emph{Social Choice and Welfare}, 65\penalty0 (4):\penalty0 983--1016, 2025.

\bibitem[Nemhauser and Wolsey(1990)]{nemhauserwolsey1990}
George~L. Nemhauser and Laurence~A. Wolsey.
\newblock A recursive procedure to generate all cuts for 0--1 mixed integer programs.
\newblock \emph{Mathematical Programming}, 46:\penalty0 379--390, 1990.

\bibitem[Neoh and Teh(2025)]{neohteh2025efx}
Tzeh~Yuan Neoh and Nicholas Teh.
\newblock Understanding {EFX} allocations: Counting and variants.
\newblock In \emph{Proceedings of the 39th AAAI Conference on Artificial Intelligence (AAAI)}, pages 14036--14044, 2025.

\bibitem[Nguyen and Rothe(2020)]{nguyenrothe2020}
Trung~Thanh Nguyen and J{\"o}rg Rothe.
\newblock Approximate {Pareto} set for fair and efficient allocation: Few agent types or few resource types.
\newblock In \emph{Proceedings of the 29th International Joint Conference on Artificial Intelligence (IJCAI)}, pages 290--296, 2020.

\bibitem[Nguyen and Rothe(2023)]{nguyenrothe2023}
Trung~Thanh Nguyen and J{\"o}rg Rothe.
\newblock Fair and efficient allocation with few agent types, few item types, or small value levels.
\newblock \emph{Artificial Intelligence}, 314:\penalty0 103820, 2023.

\bibitem[Plaut and Roughgarden(2020)]{plautroughgarden2020}
Benjamin Plaut and Tim Roughgarden.
\newblock Almost envy-freeness with general valuations.
\newblock \emph{SIAM Journal on Discrete Mathematics}, 34\penalty0 (2):\penalty0 1039--1068, 2020.

\bibitem[Scarlett et~al.(2023)Scarlett, Teh, and Zick]{scarlett2023oneandall}
Jonathan Scarlett, Nicholas Teh, and Yair Zick.
\newblock For one and all: Individual and group fairness in the allocation of indivisible goods.
\newblock In \emph{Proceedings of the 22nd International Conference on Autonomous Agents and Multiagent Systems (AAMAS)}, pages 2466--2468, 2023.
\newblock Extended abstract. Full version: arXiv:2302.06958.

\bibitem[Suksompong and Teh(2022)]{suksompongteh2022weightednash}
Warut Suksompong and Nicholas Teh.
\newblock On maximum weighted {Nash} welfare for binary valuations.
\newblock \emph{Mathematical Social Sciences}, 117:\penalty0 101--108, 2022.

\bibitem[Suksompong and Teh(2023)]{suksompongteh2023matroid}
Warut Suksompong and Nicholas Teh.
\newblock Weighted fair division with matroid-rank valuations: Monotonicity and strategyproofness.
\newblock \emph{Mathematical Social Sciences}, 126:\penalty0 48--59, 2023.

\bibitem[Zhang(2026)]{zhang2026}
Yuhao Zhang.
\newblock Exact {MMS} allocations under personalized bivalued valuations: Goods and chores, 2026.
\newblock arXiv:2608.15822.

\end{thebibliography}

\clearpage
\appendix

\section{Target Partitions and MMS Formulas for the Lower Valuations}\label{app:partitions}
This section proves the partition conditions used for the lower-valued goods and derives explicit MMS formulas. The conditions are necessary as well as sufficient; preserving them will therefore preserve an entire target partition. We fix a single valuation $v$ and use the normalization and target rounding from Section~\ref{sec:preliminaries}, so all singleton values and targets are integers. Zero-valued goods can be assigned arbitrarily. We write $(z)_+=\max\{z,0\}$.

The proofs use standard convexity arguments for balancing bundle counts and greedy allocation by decreasing marginal reduction in a deficit \citep{federgruengroenevelt1986}. We show that these choices minimize the required deficits simultaneously, and derive necessary and sufficient conditions for the three lower families.

\subsection{Values \texorpdfstring{$\{0,1,d,d+1\}$}{0, 1, d, d+1}}\label{app:partition-adj}

The normalized singleton values lie in $\{0,1,d,d+1\}$ for a fixed integer $d\ge2$. We call goods worth $d$ or $d+1$ \emph{large}. Let $a,b,c$ be the numbers of goods worth $d+1,d,1$, respectively, and let $x=a+b$
be the number of large goods.

\begin{lemma}\label{lem:partition-adj}
Let $k\ge1$ and $t\ge1$ be integers. Let $t=d\lambda+\rho=(d+1)\eta+\sigma$ and $0\le\rho<d,\quad 0\le\sigma\le d$.
The goods can be partitioned into $k$ bundles of value at least $t$ if and only if
\[
    dx+a+c\ge kt,\quad
    \rho x+a+c\ge k\rho(\lambda+1),\quad
    \sigma x+c\ge k\sigma(\eta+1).
\]
When $\rho=0$ or $\sigma=0$, the corresponding inequality is automatic.
\end{lemma}

The first inequality requires enough total value. The other two account for the fact that large goods cannot be divided: if a bundle contains too few of them, it needs goods worth one or the extra unit from goods worth $d+1$ to reach the target. The proof shows that no further conditions are needed.

\begin{proof}
We first find the minimum number of goods worth one needed for given counts of large goods, and then choose those counts. Prescribe that bundle $j$ receives $N_j$ large goods, where $N_j\ge0$ and $\sum_{j=1}^kN_j=x$. For $p\in\{d,d+1\}$, define $D_p=\sum_{j=1}^k\pos{t-pN_j}$.
If all large goods are first counted as worth $d$, the bundles need $D_d$ further units. Each of the $a$ goods worth $d+1$ can reduce this amount by one. In bundle $j$, exactly $\min\{N_j,\pos{t-dN_j}\}$ large-good positions can reduce a positive deficit. Summed over all bundles, this number is $D_d-D_{d+1}$; after all such positions are used, the remaining deficit is $D_{d+1}$. Assigning goods worth $d+1$ to these positions first therefore reduces $D_d$ by $\min\{a,D_d-D_{d+1}\}$. Hence, for the prescribed numbers $N_j$, the minimum number of goods worth one that is needed equals
\begin{equation}\label{eq:adj-ones-fixed}
    \max\{D_d-a,D_{d+1}\}.
\end{equation}

For each $p\in\{d,d+1\}$, the function $y\mapsto\pos{t-py}$ is convex on the integers. If two bundle counts differ by at least two, moving one large good from the larger count to the smaller count cannot increase the sum of these functions. Therefore both $D_d$ and $D_{d+1}$ are minimized by the same choice: distribute the $x$ large goods so that the numbers $N_j$ differ by at most one. Since this one choice minimizes both terms in \eqref{eq:adj-ones-fixed}, it also minimizes their maximum.

For this balanced distribution,
\begin{align*}
    \min D_d
      &=\max\{kt-dx,\ \rho[k(\lambda+1)-x],\ 0\},\\
    \min D_{d+1}
      &=\max\{kt-(d+1)x,\ \sigma[k(\eta+1)-x],\ 0\}.
\end{align*}
For example, the first expression equals $kt-dx$ while $x\le k\lambda$, equals $\rho[k(\lambda+1)-x]$ while $k\lambda\le x\le k(\lambda+1)$, and is zero thereafter. The second expression follows in the same way.

Since $a\le x$, we have that
\[
    kt-(d+1)x\le kt-dx-a.
\]
Thus the minimum number of unit-valued goods needed is
\[
\max\{kt-dx-a,\ \rho[k(\lambda+1)-x]-a,\
        \sigma[k(\eta+1)-x],\ 0\}.
\]
This number is at most $c$ exactly when the three inequalities in Lemma~\ref{lem:partition-adj} hold. Any unused goods can be added arbitrarily, which proves the statement for partitions of all goods.
\end{proof}

The same balancing argument gives the MMS directly, without searching over target values.

\begin{corollary}\label{cor:mms-adj}
For the counts $a,b,c$ of goods in $U$, let $x=a+b$, $\ell=\floor{\frac{x}{k}}$, and $h=k-(x-k\ell)$.
Here $h$ counts the bundles receiving $\ell$ large goods when the counts differ by at most one. Thus $1\le h\le k$, and
\[
\boxed{
\MMS_k(v,U)=\min\left\{
\floor{\frac{(d+1)a+db+c}{k}},\quad
 d\ell+\floor{\frac{a+c}{h}},\quad
 (d+1)\ell+\floor{\frac{c}{h}}
\right\}.}
\]
\end{corollary}

\begin{proof}
In a balanced distribution of the $x$ large goods, exactly $h$ bundles receive $\ell$ large goods and the other $k-h$ bundles receive $\ell+1$. For every positive integer $p$,
\[
 h\pos{t-p\ell}+(k-h)\pos{t-p(\ell+1)}
 =\max\{kt-px,\ h(t-p\ell),\ 0\}.
\]
The proof of Lemma~\ref{lem:partition-adj} shows that a target $t$ is feasible exactly when the expression with $p=d$ is at most $a+c$ and the expression with $p=d+1$ is at most $c$. These two requirements give
\[
\begin{aligned}
 t&\le\floor{\frac{dx+a+c}{k}},
 &t&\le d\ell+\floor{\frac{a+c}{h}},\\
 t&\le\floor{\frac{(d+1)x+c}{k}},
 &t&\le(d+1)\ell+\floor{\frac{c}{h}}.
\end{aligned}
\]
The third bound is no smaller than the first because $dx+a+c\le(d+1)x+c$. The greatest feasible integer $t$ is therefore the stated minimum.
\end{proof}

\subsection{Values \texorpdfstring{$\{0,1,2,2e\}$}{0, 1, 2, 2e}}\label{app:partition-even}
This family is factored, so its MMS can already be computed in polynomial time by \citet{ebadian2022}. We give explicit partition conditions because the removal arguments below need them.

Put $H=2e\ge4$. Let $a,b,c$ be the counts of goods worth $H,2,1$, respectively, and let $x=a+b$, $V=Ha+2b+c$.

\begin{lemma} \label{lem:partition-even}
Let $k\ge1$ and $t\ge0$ be integers. Let
\[
 t=H\lambda+r,\quad 0\le r<H,\quad J=\lambda+1,\quad
 R=\ceil{r/2},\quad s=\ceil{t/2}=e\lambda+R.
\]
A $k$-partition at target $t\ge0$ exists if and only if
\begin{align*}
 V&\ge kt,\\
 W_1:=ra+2b+c&\ge krJ,\\
 W_2:=ea+b+c&\ge ks,\\
 W_3:=Ra+b+c&\ge kRJ.
\end{align*}
\end{lemma}
\begin{proof}
Prescribe $N_j$ goods of value $H$ in bundle $j$, where $\sum_jN_j=a$. Its deficit is $D_j=\pos{t-HN_j}$. Distributing the $b$ two-valued goods minimizes the total remaining deficit to
\[
 \max\left\{\sum_jD_j-2b,\quad\sum_j\ceil{D_j/2}-b,\quad0\right\}.
\]
Indeed, two-valued goods have marginal deficit reductions two, then possibly one, then zero; there are $\sum_j\floor{D_j/2}$ reductions of two and $\sum_j\ceil{D_j/2}$ positive reductions in total. Since $H=2e$, $\ceil{D_j/2}=\pos{s-eN_j}$.
Both sums are discrete convex functions of the $N_j$, so balancing the $N_j$ minimizes them simultaneously: transferring a good of value $H$ from a group with at least two more such goods to a smaller group cannot increase either sum. For balanced counts, direct evaluation gives
\begin{align*}
 \sum_jD_j=\max\{kt-Ha,\ r(kJ-a),\ 0\}, \quad 
 \sum_j\ceil{D_j/2}=\max\{ks-ea,\ R(kJ-a),\ 0\}.
\end{align*}
The second identity remains valid when $R=e$, since its two nonzero expressions then coincide. Requiring the two sums to be at most $2b+c$ and $b+c$, respectively, proves the four stated inequalities.
\end{proof}

As in the preceding family, the balanced distribution also gives the MMS directly.

\begin{corollary} \label{cor:mms-even}
With $\ell=\floor{a/k}$ and $h=k-(a-k\ell)$,
\[
 \MMS_k(v,U)=\min\left\{
 \floor{V/k}, \ H\ell+\floor{(2b+c)/h},\
 2\floor{(ea+b+c)/k},\ H\ell+2\floor{(b+c)/h}
 \right\}.
\]
\end{corollary}
\begin{proof}
The two balanced sums in the proof of Lemma~\ref{lem:partition-even} equal
\[
 \max\{kt-Ha,h(t-H\ell),0\},\quad
 \max\{ks-ea,h(s-e\ell),0\}.
\]
Solving the four resulting bounds for the largest integral $t$ gives the stated MMS formula.
\end{proof}

\subsection{Values \texorpdfstring{$\{0,1,2,3,4\}$}{0, 1, 2, 3, 4}}\label{app:partition-five}
For this family, use the cumulative counts
\[
 P=\#\{g:v(g)\ge1\},\quad Z=\#\{g:v(g)\ge2\},\quad
 Y=\#\{g:v(g)\ge3\},\quad A=\#\{g:v(g)=4\}.
\]
Thus $P\ge Z\ge Y\ge A$ and $V=v(U)=P+Z+Y+A$.

The partition conditions and the MMS formula use the following weights. For integers $p\ge2$ and $1\le r<p$, define $w_{p,r}(z)=r\floor{z/p}+\min\{r,z\bmod p\}$.

\begin{lemma} \label{lem:partition-five}
For integers $k\ge1$ and $t\ge0$, the goods have a partition into $k$ bundles each of value at least $t$ if and only if $V\ge kt$ and all applicable rows of the following table hold:
\begin{center}
\begin{tabular}{lll}
\toprule
Condition on $t$ & Required inequality & Definition of $J$\\
\midrule
$t$ odd & $P+Y\ge k\ceil{t/2}$ & ---\\
$t\equiv1\pmod3$ & $P+A\ge kJ$ & $t=3J-2$\\
$t\equiv2\pmod3$ & $P+Z+A\ge2kJ$ & $t=3J-1$\\
$t\equiv1\pmod4$ & $P\ge kJ$ & $t=4J-3$\\
$t\equiv2\pmod4$ & $P+Z\ge2kJ$ & $t=4J-2$\\
$t\equiv3\pmod4$ & $P+Z+Y\ge3kJ$ & $t=4J-1$\\
\bottomrule
\end{tabular}
\end{center}
For each of the moduli two, three, and four, at most one row applies. Thus a fixed target requires total value and at most three additional inequalities.
\end{lemma}

\begin{proof}
We first prove necessity using the weights above, then construct a partition from the stated inequalities.

\paragraph{Necessity.}
The weights $w_{p,r}$ come from the standard rounding inequalities for integer programming \citep{nemhauserwolsey1990,atamturkgunluk2010}. The function is nondecreasing and subadditive on nonnegative integers; we give the verification for completeness. To check subadditivity, reduce both arguments modulo $p$. Without a carry it is the subadditivity of $\min\{r,\cdot\}$. With a carry, for residues $a+b\ge p$, the required inequality is
\[
 r+\min\{r,a+b-p\}\le\min\{r,a\}+\min\{r,b\}.
\]
If both residues are at least $r$, the right side is $2r$; if one is below $r$, the inequality follows from the other residue being at most $p-1$. If both are below $r$, the right side is $a+b$, and $r+a+b-p\le a+b$. These cases prove the claim. Therefore every bundle of value at least $t$ has total $w_{p,r}$-weight at least $w_{p,r}(t)$. For $r=t\bmod p>0$, this is $r\ceil{t/p}$. Evaluating these weights on values $1,2,3,4$, for $p=2,3,4$, gives the displayed rows.

\paragraph{Sufficiency: minimizing the number of ones.}
Let $a,b,c,u$ count fours, threes, twos, and ones, respectively, so $Z=a+b+c$, $Y=a+b$, $A=a$, and $P=Z+u$. Prescribe balanced numbers $N_1,\dots,N_k$ of nonunit positive goods: they sum to $Z$ and differ by at most one. Initially value every such position at two. Let
\[
 E_p=\sum_{j=1}^k\pos{t-pN_j}\quad \text{for }p=2,3,4,
 \quad
 H=\sum_{j=1}^k\max\{t-3N_j,\ceil{t/2}-N_j,0\}.
\]
We first establish that the minimum number of ones needed, for these prescribed positions, is
\begin{equation}\label{eq:five-minones}
 \max\{E_2-2a-b,\ H-a-b,\ E_3-a,\ E_4\}.
\end{equation}

If bundle $j$ receives $x_j$ fours, its remaining deficit before threes is $C_j(x_j)=\pos{t-2N_j-2x_j}$.
At most $N_j-x_j$ threes can be placed there. After distributing all $b$ threes, the least remaining total deficit is
\[
 \max\left\{\sum_j C_j(x_j)-b,\ \sum_j G_j(x_j)\right\},
 \quad G_j(x_j)=\pos{t-3N_j-x_j}.
\]
Indeed, each useful three reduces the deficit by one, and there are exactly $C_j(x_j)-G_j(x_j)$ useful available positions in group $j$.

A successive four has one of the marginal benefit pairs
\[
 (\Delta C_j,\Delta G_j)=(2,1),(2,0),(1,0),(0,0),
\]
in that order. A decrease in $G_j$ always gives a decrease of two in $C_j$: if $x_j<N_j$ and $t-3N_j-x_j\ge1$, then $t-2N_j-2x_j\ge2$. Thus assigning fours in the displayed priority order simultaneously minimizes both sums. If $F=\sum_j\min\{N_j,\floor{\pos{t-2N_j}/2}\}$,
then $H=E_2-F$, and the two minima are
\[
 \min\sum_j C_j=\max\{E_2-2a,H-a,E_4\},\quad
 \min\sum_j G_j=\max\{E_3-a,E_4\}.
\]
Here $F$ counts the reductions of two in the total deficit, $H=E_2-F$ accounts for the subsequent reductions of one, and $E_4$ is the deficit when every position is a four. This proves \eqref{eq:five-minones}. Any excess fours or threes occupy unused positions; there are enough because $a+b\le Z$.

\paragraph{Sufficiency: bounding the deficits.}
It remains to show that every term in \eqref{eq:five-minones} is at most $u$. For balanced $N_j$, with $r_p=t\bmod p$, direct evaluation at the two possible sizes gives
\begin{equation}\label{eq:five-Ep}
 E_p=\max\{kt-pZ,\ r_p[k\ceil{t/p}-Z],\ 0\}.
\end{equation}
Also, putting $s=\ceil{t/2}$ and $h=\ceil{3t/4}$, the same evaluation gives
\begin{equation}\label{eq:five-H}
 H=\max\{kt-3Z,\ kh-2Z,\ ks-Z,\ 0\}.
\end{equation}
For completeness, the summand defining $H$ has successive discrete slopes $-3$, possibly $-2$, $-1$, and $0$. The line of slope $-2$ is needed only when $t\equiv2,3\pmod4$, and its intercept is $\ceil{3t/4}$. This proves \eqref{eq:five-H} for balanced sizes, including a size zero.

Total value gives $kt-2Z\le2a+b+u$. If $t$ is odd, the parity row gives $k\ceil{t/2}-Z\le a+b+u$. Hence \eqref{eq:five-Ep} gives $E_2\le2a+b+u$.

Since $V=3Z+a-c+u\le3Z+a+u$, total value gives $kt-3Z\le a+u$. Each applicable modulus-three row is precisely
$r_3[k\ceil{t/3}-Z]\le a+u$. Thus $E_3\le a+u$.

Since $V\le4Z+u$, total value gives $kt-4Z\le u$. The modulus-four rows give $r_4[k\ceil{t/4}-Z]\le u$; for $r_4=3$ the stated row is even stronger by $c$. Thus $E_4\le u$.

To bound $H$, its first term is at most $a+u$. Its third term is at most $a+b+u$, by the parity row when $t$ is odd and by half the total-value inequality when $t$ is even. When $t\equiv0,1\pmod4$, the second term in \eqref{eq:five-H} is the average of the first and third. When $t=4\ell+2$, it is the average of $kt-2Z$ and $2k(\ell+1)-2Z$, which are at most $2a+b+u$ and $u$, respectively. When $t=4\ell+3$, the modulus-four row says directly that $3k(\ell+1)-2Z\le a+b+u$. Hence $H\le a+b+u$ in every case. Formula \eqref{eq:five-minones} is at most $u$, so the required partition exists.
\end{proof}

\begin{corollary} \label{cor:mms-five}
For $p=2,3,4$ and $1\le r<p$, set $W_{p,r}=\sum_g w_{p,r}(v(g))$, $M_{p,r}=\floor{W_{p,r}/k}$, and $L_{p,r}=p\floor{M_{p,r}/r}+(M_{p,r}\bmod r)$.
Then
\[
 \MMS_k(v,U)=\min\left\{\floor{V/k},\ L_{2,1},L_{3,1},L_{3,2},
 L_{4,1},L_{4,2},L_{4,3}\right\}.
\]
The six weights, in the order appearing here, are
$P+Y$, $P+A$, $P+Z+A$, $P$, $P+Z$, and $P+Z+Y$.
\end{corollary}
\begin{proof}
Every target partition satisfies $W_{p,r}\ge k w_{p,r}(t)$ for every fixed $(p,r)$ by the necessity argument above. The largest nonnegative integer $t$ with $w_{p,r}(t)\le M$ is $p\floor{M/r}+(M\bmod r)$. Thus the displayed minimum is an upper bound on MMS. Conversely, taking $t$ equal to this minimum satisfies total value and, in particular, all the applicable inequalities of Lemma~\ref{lem:partition-five}. Sufficiency proves equality.
\end{proof}

\section{Proof of Lemma~\ref{lem:lower-properties}}

\subsection{Proof of Lemma~\ref{lem:lower-properties}: Property \ref{item:one}}\label{app:one}
We prove property~\ref{item:one} from Section~\ref{sec:extension} for the lower valuations. Throughout, $v$ is ordered, $k\ge2$, and the positive integral target $t$ admits a $k$-partition of $U$. Let $q_v$ be the first acceptable spaced size. Lemma~\ref{lem:rank-reduction} first gives an exchange valid for all additive valuations. Lemma~\ref{lem:low-removal} then treats bundles worth at most $t$ for our three lower families. The final subsection uses the rightmost comparison to reduce property~\ref{item:one} to these two lemmas.

\subsubsection{An exchange for consecutive ranks}
For $1\le q\le q_v$, define the consecutive bundle $R_{q-1}=\{u_{(q-1)(k-1)+1},\dots,u_{(q-1)k+1}\}$.
It contains $q$ goods and belongs to $\mathcal F_{k,q}$.

The following lemma uses the standard pigeonhole and exchange argument for rank-based reductions \citep{gargtaki2021,heidari2026}. We state it for arbitrary bundle size.

\begin{lemma}\label{lem:rank-reduction}
Suppose an agent can partition $U$ into $k\ge2$ bundles of value at least $t$. After removing $R_{q-1}$, the remaining goods can be partitioned into $k-1$ bundles of value at least $t$.
\end{lemma}

\begin{proof}
Add any required zero-valued dummy goods to the given partition. Among the first $(q-1)k+1$ goods, some bundle contains at least $q$ goods. List any $q$ such goods in rank order. Its $j$-th listed good has rank at most $(q-1)(k-1)+j$, which is the rank of the $j$-th good of $R_{q-1}$. Each listed good is therefore worth at least as much as the corresponding good of $R_{q-1}$. Cancel any goods common to the two sets and match the remaining goods in rank order. Move the goods of $R_{q-1}$ into the chosen partition bundle and use the matched listed goods to replace them in the other bundles; no other bundle decreases in value. Delete the chosen bundle and distribute its remaining goods among the other $k-1$ bundles. Discarding dummy goods does not change any value.
\end{proof}

\subsubsection{Bundles worth at most the target}
\begin{lemma}\label{lem:low-removal}
Suppose $v$ takes singleton values in one of the sets in \eqref{eq:lower-values}. If $B\in\mathcal F_{k,q}$ for $q\le q_v$, $v(R_{q-1})<t$, and $v(B)\le t$, then $U\setminus B$ admits a $(k-1)$-partition at $t$.
\end{lemma}
\begin{proof}
For $k=2$, the remaining total value is at least $t$, so assume $k\ge3$. We verify the exact partition conditions from Appendix~\ref{app:partitions} separately for each family.

\subsubsection*{Values $\{0,1,d,d+1\}$}

Use the counts $a,b,c$ and $x=a+b$ from Appendix~\ref{app:partition-adj}. Let $t=d\lambda+\rho=(d+1)\eta+\sigma$ and $\quad 0\le\rho<d,\quad 0\le\sigma\le d$.
Let $a_B$ be the number of goods worth $d+1$ in $B$, let $x_B$ be the number of large goods in $B$, and let $c_B$ be the number worth one.

The selected ranks $r_1<\cdots<r_q$ satisfy $r_j\ge(j-1)k+1$ for all $j = 1,\dots, q$.
Consequently, at most $\ceil{X/k}$ selected goods lie in any prefix of $X$ goods. Define $\bar a=\ceil{\frac ak}$ and $\bar x=\ceil{\frac xk}$.
Since the goods worth $d+1$ and the large goods form prefixes, and since $q\le q_v$, we have $a_B\le\bar a$, $x_B\le\bar x$, and $x_B+c_B\le q_v$.
The value assumption gives us
\begin{equation}\label{eq:adj-removed-value}
    dx_B+a_B+c_B=v(B)\le t.
\end{equation}
We verify the three inequalities of Lemma~\ref{lem:partition-adj} for the remaining counts. The total-value inequality follows immediately from \eqref{eq:adj-removed-value}.

We will repeatedly use a consequence of this count bound: for nonnegative integers $X,J$ with $X\ge kJ$, at least $X-\ceil{X/k}\ge(k-1)J$ goods remain in a prefix of length $X$. We now distinguish whether the agent needs goods worth one to reach her target in the spaced bundle.

\paragraph{The spaced target uses unit-valued goods.}
Suppose
\begin{equation}\label{eq:adj-needs-ones}
    d\bar x+\bar a<t.
\end{equation}
The full spaced sequence contains $\bar a$ goods worth $d+1$ and $\bar x$ large goods. Every subsequent positive good in that sequence is worth one. Therefore
\begin{equation}\label{eq:adj-qv}
    q_v=t-(d-1)\bar x-\bar a.
\end{equation}

Assume $\rho>0$. Use the original total-value inequality from Lemma~\ref{lem:partition-adj}, bound $c_B$ by $q_v-x_B$, and substitute \eqref{eq:adj-qv}. Rearranging gives us
\begin{align*}
&\rho(x-x_B)+(a-a_B)+(c-c_B) \\
&\quad\ge
(k-1)[t-(d-\rho)\bar x]
 +(d-\rho)(k\bar x-x)
 +(\rho-1)(\bar x-x_B)
 +(\bar a-a_B).
\end{align*}
Every term after the first is nonnegative. Moreover, \eqref{eq:adj-needs-ones} implies $\bar x\le\lambda$, and hence $t-(d-\rho)\bar x\ge\rho(\lambda+1)$.
Thus the second partition inequality holds after the allocation with $k-1$ in place of $k$.

Assume $\sigma>0$. The same substitution and rearrangement give
\begin{align*}
&\sigma(x-x_B)+(c-c_B) \\
&\quad\ge
(k-1)[t-(d-\sigma)\bar x-\bar a]
 +(d-\sigma)(k\bar x-x)
 +(k\bar a-a)
 +(\sigma-1)(\bar x-x_B).
\end{align*}
Again, every term after the first is nonnegative. If $\bar x\le\eta$, then $\bar a\le\bar x$ and
\[
 t-(d-\sigma)\bar x-\bar a
 \ge t-(d+1-\sigma)\eta
 =\sigma(\eta+1).
\]
If $\bar x\ge\eta+1$, then \eqref{eq:adj-needs-ones} gives
\[
 t-(d-\sigma)\bar x-\bar a
 =t-d\bar x-\bar a+\sigma\bar x
 \ge\sigma(\eta+1).
\]
Therefore the third partition inequality is also preserved.

\paragraph{The spaced target uses only large goods.}
Now suppose $d\bar x+\bar a\ge t$. Then $q_v\le\bar x$.

First consider the third partition inequality, assuming $\sigma>0$. If $\bar x\le\eta+1$, the removed amount from its left-hand side is at most
\[
    \sigma x_B+c_B
    \le\sigma(x_B+c_B)
    \le\sigma q_v
    \le\sigma(\eta+1).
\]
If $\bar x\ge\eta+2$, then the prefix count bound gives $x-x_B\ge(k-1)(\eta+1)$,
which proves the required inequality directly.

It remains to consider the second partition inequality, assuming $\rho>0$. If $x-x_B\le(k-1)\lambda$,
the new total-value inequality gives
\[
\rho(x-x_B)+(a-a_B)+(c-c_B)
\ge(k-1)t-(d-\rho)(x-x_B)
\ge(k-1)\rho(\lambda+1).
\]
If $x-x_B\ge(k-1)(\lambda+1)$, the first term already gives the required bound. We may therefore assume
\begin{equation}\label{eq:adj-middle-count}
    (k-1)\lambda<x-x_B<(k-1)(\lambda+1).
\end{equation}
Equation~\eqref{eq:adj-removed-value} implies $x_B\le\lambda$. If $x_B=\lambda$, it also gives $a_B+c_B\le\rho$, so the removed amount from the left-hand side of the second partition inequality is at most $\rho(\lambda+1)$.

Suppose $x_B\le\lambda-1$. If $a\ge k\rho$, the prefix count bound gives $a-a_B\ge(k-1)\rho$, which together with \eqref{eq:adj-middle-count} proves the desired bound. We may therefore assume $a<k\rho$, and hence $\bar a\le\rho$. Equation~\eqref{eq:adj-middle-count} and the same prefix count bound imply $\bar x\le\lambda+1$. Thus
\[
    x_B+c_B\le q_v\le\bar x\le\lambda+1.
\]
If $x_B+c_B=\lambda+1$, then $a_B\le\rho-1$. Indeed, otherwise $a_B=\bar a=\rho\le x_B\le\lambda-1$ and $q_v=\bar x=\lambda+1$. The first $\lambda$ goods of the spaced sequence would then contain $\rho$ goods worth $d+1$ and $\lambda-\rho$ goods worth $d$, and hence would have value $d\lambda+\rho=t$, contradicting the definition of $q_v$. Consequently, in all cases, $a_B+x_B+c_B\le\rho+\lambda$.
The amount removed from the left-hand side of the second partition inequality is therefore at most
\[
\rho x_B+a_B+c_B
=(\rho-1)x_B+(a_B+x_B+c_B)
\le(\rho-1)(\lambda-1)+\rho+\lambda
=\rho\lambda+1
\le\rho(\lambda+1).
\]
This proves the last required inequality. Lemma~\ref{lem:partition-adj} now gives a partition of the remaining goods into $k-1$ bundles of value at least $t$.

\subsubsection*{Values $\{0,1,2,2e\}$}
Use $H=2e$ and the counts, weights, and target remainders from Appendix~\ref{app:partition-even}; in particular, $x=a+b$. Put $\alpha=\ceil{a/k}$ and $\zeta=\ceil{x/k}$. Let $a_B,x_B,p_B$ count the goods in $B$ worth $H$, at least two, and at least one, respectively. Then
\[
 a_B\le\alpha,\quad x_B\le\zeta,\quad p_B\le q\le q_v,\quad
 p_B+x_B+(H-2)a_B\le t.
\]
Set $h_a=k\alpha-a\ge0$ and $h_x=k\zeta-x\ge0$, and write primes for weights of the remaining goods. Total value is preserved. If $r=0$, the total-value inequality implies the $W_2$ inequality, while the $W_1$ and $W_3$ inequalities are trivial. Hence suppose $r>0$. When $r=1$, we have $R=1$ and $W_1'\ge W_3'$, with the same required lower bound; it suffices to verify $W_1'$ separately only for $r\ge2$.

\paragraph{The spaced target is reached using only high-valued goods.}
If $t\le H\alpha$, then $q_v=J\le\alpha$. The $q$ goods of $R_{q-1}$ are all worth $H$. Its rejection therefore excludes $q=J$, so $q\le\lambda$. The removed $W_2$-weight is at most $eq\le e\lambda\le s$; the removed $W_3$-weight is at most $Rq\le R\lambda<RJ$. For $r\ge2$, each good has $W_1$-weight at most $r$, so the removed $W_1$-weight is at most $rq\le r\lambda<rJ$. All inequalities survive.

\paragraph{The spaced target is reached using twos, without ones.}
Suppose $H\alpha<t\le H\alpha+2(\zeta-\alpha)$.
Then $\alpha\le\lambda$ and $q_v=s-(e-1)\alpha$. The removed $W_2$-weight is
$p_B+(e-1)a_B\le s$, preserving the $W_2$ inequality. From the original $W_2\ge ks$ and $p_B\le q_v$,
\begin{align*}
 W_3'&\ge(k-1)[s-(e-R)\alpha]
       +(e-R)h_a+(R-1)(\alpha-a_B)\\
   &\ge(k-1)RJ.
\end{align*}
The last step uses $s-(e-R)\alpha=RJ+(e-R)(\lambda-\alpha)$.

For $r\ge2$, put $\epsilon=t\bmod2$, $D_q=q_v-p_B\ge0$, and $c_*=p_B-x_B\ge0$. Using the original $V\ge kt$ gives
\begin{align*}
 W_1'\ge{}&(k-1)[t-(H-r)\alpha]-\epsilon
       +(H-r)h_a +(r-2)(\alpha-a_B)+2D_q+c_*.
\end{align*}
The first bracket exceeds $rJ$ by $(H-r)(\lambda-\alpha)$. If $\epsilon=0$, all remaining terms are nonnegative. If $\epsilon=1$, then $r$ is odd and $r\ge3$. A violation is possible only if every nonnegative integral margin vanishes, forcing
\[
 \alpha=\lambda,\quad h_a=0,\quad a_B=\alpha,\quad D_q=0,\quad c_*=0.
\]
But then $p_B=x_B=\lambda+R$, $a_B=\lambda$, and
$v(B)=H\lambda+2R=t+1$, contrary to the hypothesis. This preserves the $W_1$ inequality.

\paragraph{The spaced target needs ones.}
Finally suppose $t>H\alpha+2(\zeta-\alpha)$. Then
\[
 \alpha\le\lambda,\quad q_v=t-\zeta-(H-2)\alpha,\quad
 \zeta+(e-1)\alpha\le t-s.
\]
Substituting $p_B\le q_v$ into the original total-value bound gives
\begin{align*}
 W_2'\ge{}&(k-1)[t-\zeta-(e-1)\alpha]
          +h_x+(e-1)h_a+(e-1)(\alpha-a_B)
       \ge(k-1)s,\\[0.3em]
 W_3'\ge{}&(k-1)[t-\zeta-(H-R-1)\alpha]
          +h_x+(H-R-1)h_a+(R-1)(\alpha-a_B)
       \ge(k-1)RJ.
\end{align*}
For the second inequality, its bracket is at least $s-(e-R)\alpha\ge RJ$. All coefficients of the remaining terms are nonnegative. For $r\ge2$, the same substitution gives
\begin{align*}
 W_1'\ge{}&(k-1)[t-(H-r)\alpha]
          +(H-r)h_a+(r-2)(\alpha-a_B)+(\zeta-x_B)
       \ge(k-1)rJ.
\end{align*}
Thus every inequality of Lemma~\ref{lem:partition-even} is preserved.

\subsubsection*{Values $\{0,1,2,3,4\}$}
Use the cumulative counts $P,Z,Y,A$ from Appendix~\ref{app:partition-five}.

Let
\[
 \bar{P}=\ceil{P/k},\quad\bar{Z}=\ceil{Z/k},\quad
 \bar{Y}=\ceil{Y/k},\quad\bar{A}=\ceil{A/k},
\]
and let $p,z,y,a$ be the corresponding cumulative counts in $B$. Rank spacing gives
\begin{equation}\label{eq:five-prefixbounds}
 p\le\bar{P},\quad z\le\bar{Z},\quad y\le\bar{Y},\quad a\le\bar{A},
 \quad p\ge z\ge y\ge a,
 \quad p+z+y+a\le t.
\end{equation}
We also have $p\le q\le q_v\le\bar{P}$, and
\begin{equation}\label{eq:five-spaced}
 v(S_r)=\min(r,\bar{P})+\min(r,\bar{Z})+\min(r,\bar{Y})+\min(r,\bar{A}).
\end{equation}
In particular, $S_{p-1}$ is rejected whenever $p\ge1$. Total value is preserved. A prefix containing at least $kJ$ goods retains at least $(k-1)J$, because $X-\ceil{X/k}\ge(k-1)J$ for $X\ge kJ$. This proves the modulus-four row with remainder one.

We first state a consequence of rounding the prefix counts. Consider an inequality
$X_1+\cdots+X_d\ge kC$, with $d=2$ or $3$, and removed counts $x_j\le\ceil{X_j/k}$. If it is violated after removal, put
\[
 \delta=\sum_j\ceil{X_j/k}-C,\quad
 D=\sum_j(\ceil{X_j/k}-x_j),\quad
 s=\sum_jX_j-kC.
\]
Then
\begin{equation}\label{eq:five-rounding}
 1\le\delta\le d-1,\quad s+D\le\delta-1.
\end{equation}
Indeed, the violated inequality says $s+D\le\delta-1$. Meanwhile, the sum of the ceiling deficits $k\ceil{X_j/k}-X_j$ equals $k\delta-s$ and is at most $d(k-1)$. These two facts exclude $\delta\ge d$. For $d=2$, a violated inequality consequently requires $\delta=1$ and $s=D=0$.

\paragraph{The parity row: $P+Y\ge kL$, $t=2L-1$}
A violation requires
\[
 \bar{P}+\bar{Y}=L+1,\quad P+Y=kL,\quad p=\bar{P},\quad y=\bar{Y}.
\]
Both ceiling deficits are positive, so $P<k\bar{P}$, and $q=q_v=\bar{P}$. If $\bar{P}=\bar{Y}$, then $t=4\bar{P}-3$; any target partition needs $P\ge k\bar{P}$ because each good is worth at most four, a contradiction.

Thus $\bar{Y}<\bar{P}$. Rejection of $S_{\bar{P}-1}$ gives us
\[
 \min(\bar{Z},\bar{P}-1)+\bar{A}\le L-2.
\]
Total value and $P+Y=kL$ give $Z+A\ge k(L-1)$. If $\bar{Z}<\bar{P}$, these statements contradict each other immediately. If $\bar{Z}=\bar{P}$, they force $Z=k\bar{P}>P$, also impossible.

\paragraph{The modulus-three row with remainder one}
Write $t=3J-2$. A violation requires
\[
 \bar{P}+\bar{A}=J+1,\quad P+A=kJ,\quad p=\bar{P},\quad a=\bar{A},
 \quad P<k\bar{P}.
\]
Here $q=q_v=\bar{P}$. If $\bar{A}=\bar{P}$, then both $B$ and $S_{q_v}$ consist entirely of fours. Hence $t=4\bar{P}$, while $J=2\bar{P}-1$ gives $t=6\bar{P}-5$, impossible for integral $\bar{P}$.

Otherwise rejection of $S_{\bar{P}-1}$ gives
\[
 \min(\bar{Z},\bar{P}-1)+\min(\bar{Y},\bar{P}-1)\le2J-3,
\]
whereas total value gives $Z+Y\ge2k(J-1)$. If neither barred count equals $\bar{P}$, the inequalities directly conflict. If only $\bar{Z}=\bar{P}$, they force $Z=k\bar{P}>P$. If both equal $\bar{P}$, they give $\bar{P}\le J-1$, so
$Z+Y\le2P<2k\bar{P}\le2k(J-1)$. All cases are impossible.

\paragraph{The modulus-four row with remainder two}
Write $t=4J-2$ and $L=2J$. A violation requires
$\bar{P}+\bar{Z}=L+1$, $P+Z=kL$, $p=\bar{P}$, and $z=\bar{Z}$.
The barred counts cannot be equal because $L+1$ is odd. Thus $\bar{Z}<\bar{P}$, and rejection of $S_{\bar{P}-1}$ gives $\bar{Y}+\bar{A}\le L-3$. But total value gives $Y+A\ge k(L-2)$, a contradiction.

\paragraph{The modulus-four row with remainder three}
Write $t=4J-1$ and suppose $W=P+Z+Y\ge3kJ$ is violated after removal. Use \eqref{eq:five-rounding} with $C=3J$. Then
\[
 \bar{P}+\bar{Z}+\bar{Y}=3J+\delta,\quad
 \delta\in\{1,2\},\quad s+D\le\delta-1.
\]
Total value implies $A\ge k(J-1)-s$. Since $s\le1$, we have $\bar{Y}\ge\bar{A}\ge J-1$.

If $\delta=1$, then $s=D=0$, $p=\bar{P}$, and the removed weight is $3J+1$, so $a\le J-2$. The sorted offsets
$(\bar{P}-J,\bar{Z}-J,\bar{Y}-J)$ sum to one and are all at least $-1$. There are exactly four possibilities. The following are lower bounds obtained from \eqref{eq:five-spaced} and $\bar{A}\ge J-1$:
\begin{center}
\begin{tabular}{ccccc}
\toprule
Offsets & $(3,-1,-1)$ & $(2,0,-1)$ & $(1,1,-1)$ & $(1,0,0)$\\
\midrule
$v(S_{\bar{P}-1})-(4J-1)$ & $0$ & $0$ & $-1$ & $0$\\
\bottomrule
\end{tabular}
\end{center}
Since $S_{\bar{P}-1}$ is rejected, only $(1,1,-1)$ remains. Then $Y\le k(J-1)\le A$, hence $Y=A$. But $y=\bar{Y}=J-1$, so $a=y=J-1$, contradicting $a\le J-2$.

If $\delta=2$, the sorted offsets sum to two and have the following possibilities:
\begin{center}
\begin{tabular}{c|ccccc}
\toprule
Offsets & $(4,-1,-1)$ & $(3,0,-1)$ & $(2,1,-1)$ & $(2,0,0)$ & $(1,1,0)$\\
\midrule
$v(S_{\bar{P}-1})-(4J-1)$ & $1$ & $1$ & $1$ & $1$ & $0$\\
$v(S_{\bar{P}-2})-(4J-1)$ & $0$ & $0$ & $-1$ & $0$ & ---\\
\bottomrule
\end{tabular}
\end{center}
The first row rules out $p=\bar{P}$. Therefore $p=\bar{P}-1$, $D=1$, $s=0$, and $z=\bar{Z}$, $y=\bar{Y}$. The last offset pattern is incompatible with $p\ge z$. The second row and rejection of $S_{p-1}$ leave only $(2,1,-1)$. Once again $Y=A=k(J-1)$, whereas the removed weight is $3J+1$ and so $a\le J-2$. Since $y=\bar{Y}=J-1$, this is impossible. The modulus-four row is preserved.

\paragraph{The modulus-three row with remainder two}
Write $t=3J-1$ and suppose $W=P+Z+A\ge2kJ$ is violated after removal. Now
\begin{equation}\label{eq:five-mod3state}
 \bar{P}+\bar{Z}+\bar{A}=2J+\delta,\quad
 \delta\in\{1,2\},\quad s+D\le\delta-1,
 \quad Y\ge k(J-1)-s.
\end{equation}
In particular $\bar{P}\ge\bar{Z}\ge\bar{Y}\ge J-1$. In fact $\bar{P}\ge J$. Otherwise $\bar{P}=\bar{Z}=J-1$, and the ceiling deficits of $P$ and $Z$ are each at most $s$. The deficit of $A$ is at most $k-1$. Their sum is $k\delta-s$, so
\[
 k(\delta-1)+1\le3s.
\]
This is impossible when $\delta=1,s=0$, and when $\delta=2,s\le1,k\ge3$.

The removed weight is $2J+\delta-D$, so
$y\le J-1-\delta+D$. Let $d_A=\bar{A}-a\le D$. Then
$\bar{A}\le y+d_A\le J-1$ (checking $\delta=1,D=0$, or $\delta=2,D\le1$).
If $p=\bar{P}$, put $e=1$ when $\bar{Z}=\bar{P}$ and $e=0$ otherwise. Equation~\eqref{eq:five-spaced} gives
\[
 v(S_{\bar{P}-1})\ge3J+\delta-2-e.
\]
Rejection is possible only for $\delta=1,e=1$. In that case
$2\bar{P}+\bar{A}=2J+1$ and $\bar{P}\ge J$, forcing
$\bar{P}=\bar{Z}=J$, $\bar{A}=1$, and $q=J$. Since $Z\ge k(J-1)+1$ and $Y\ge k(J-1)$, the rank reduction $R_{J-1}$ contains $J$ goods worth at least two, of which at least $J-1$ are worth at least three. Its value is at least $3J-1$, contradicting rejection.

Thus $p=\bar{P}-1$, which forces $D=1$, $\delta=2$, $s=0$, $z=\bar{Z}$, and $a=\bar{A}$. Also $\bar{P}>\bar{Z}$. If $\bar{P}=J$, then $\bar{Z}=J-1$ and $Y=Z=k(J-1)$. Equation~\eqref{eq:five-mod3state} gives $\bar{A}=3$, but $P\ge k(J-1)+1$ and $P+Z+A=2kJ$ give $A\le2k-1$, a contradiction. Hence $\bar{P}\ge J+1$.

Now $a=\bar{A}\le y\le J-2$. Put $e'=1$ when $\bar{Z}=\bar{P}-1$ and $e'=0$ otherwise. Equation~\eqref{eq:five-spaced} gives
\[
 v(S_{\bar{P}-2})\ge3J-1-e'.
\]
Rejection requires $e'=1$. Hence $2\bar{P}-1+\bar{A}=2J+2$, forcing
$\bar{P}=J+1$, $\bar{Z}=J$, and $\bar{A}=1$. But $v(S_J)\ge3J>t$, so $p=J\le q\le q_v\le J$. Thus $q=J$, and the same count argument makes $R_{J-1}$ acceptable, a final contradiction.

Every inequality of Lemma~\ref{lem:partition-five} is preserved, completing the proof.

\end{proof}

\subsubsection{Why the rightmost choice permits removal}
\begin{proof}[Proof of Lemma~\ref{lem:lower-properties}, property~\ref{item:one}]
Let $B\in\mathcal F_{k,q}$, where $q\le q_v$, and suppose its rank vector is lexicographically at least that of every $t$-acceptable member of $\mathcal F_{k,q}$. If $R_{q-1}$ is acceptable, this comparison implies that every coordinate of $B$ is at least the corresponding coordinate of $R_{q-1}$. Indeed, $R_{q-1}$ is consecutive, so once a coordinate of a strictly increasing rank vector becomes larger, no later coordinate can be smaller. The complement of $B$ is therefore at least as valuable, coordinate by coordinate, as the complement of $R_{q-1}$. Lemma~\ref{lem:rank-reduction} applies.

We may hence assume $R_{q-1}$ is rejected. A singleton is always safely removable, so take $q\ge2$. Lemma~\ref{lem:low-removal} applies if $v(B)\le t$. We only need to analyze $v(B)>t$.

Two elementary observations will be used repeatedly. First, $q=q_v$, and
\begin{equation}\label{eq:one-first-q}
 v(S_{q-1})<t<v(B)\le v(S_q).
\end{equation}
Second, a selected good of value zero or one cannot occur in $B$: removing it would leave an acceptable $(q-1)$-set dominated by $S_{q-1}$. More generally, replacing a selected good by an unselected later good with value loss at most $v(B)-t$ contradicts the assumed lexicographic comparison. Such a rightward replacement preserves membership in $\mathcal F_{k,q}$.

\paragraph{Values $\{0,1,d,d+1\}$.}
Let $w$ be the last value in $S_q$. It cannot be $d+1$: then the rank reduction consists of $q$ goods worth $d+1$ and is acceptable. If $w\le1$, \eqref{eq:one-first-q} and integrality give $v(B)\le t$. Thus $v(B)>t$ requires $w=d$ and, by the preceding observation, all selected goods have value $d$ or $d+1$.

If no selected good has value $d+1$, the consecutive bundle $R_{q-1}$ is worth at least $v(B)$. Otherwise no unselected $d$-good can exist, because swapping it for a selected $(d+1)$-good loses only one. Thus the rank reduction, which lies entirely within the large-good prefix, has no more $d$-goods than $B$ and again is worth at least $v(B)$. Both conclusions contradict rejection. Hence $v(B)\le t$, and Lemma~\ref{lem:low-removal} proves the claim.

\paragraph{Values $\{0,1,2,2e\}$.}
Put $H=2e$. The last spaced value cannot be $H$, since then the rank reduction is acceptable. If it is at most one, there is no excess above $t$. Otherwise it is two. The excess above $t$ must be exactly one, and $B$ has exactly the same value pattern as $S_q$; in particular it contains a two. Any remaining one would be unselected and could replace this two, reaching $t$ farther to the right. Hence no ones remain.

All remaining bundle values are then even. Replace $t$ by $\tau=2\ceil{t/2}$. All feasibility and acceptance statements are unchanged, and the possible excess gives $v(B)=\tau$. Lemma~\ref{lem:low-removal} at $\tau$ applies. It also guarantees the original target $t$.

\paragraph{Values $\{0,1,2,3,4\}$.}
Assume $v(B)>t$. The last spaced value $w$ is at least two and at most three: one cannot cause an excess above the target, and four would make the rank reduction acceptable. Put $\Delta=v(S_q)-v(B)$. Equation~\eqref{eq:one-first-q} gives us
\[
 \Delta+(v(B)-t)\le w-1.
\]

If $w=3$, then $\Delta\le1$. Consequently $B$ consists of threes and fours, possibly with a single three replaced by a two. The rank reduction consists entirely of threes and fours. If $B$ has no four, the rank reduction is at least as valuable. If $B$ has a four, every three in the instance must be selected, since otherwise a four-to-three swap loses one and stays acceptable. The rank reduction can then contain no more threes than $B$, and is again at least as valuable. This contradicts rejection.

Thus $w=2$, $v(B)=t+1$, and $\Delta=0$. Write $P,Y,A$ for the numbers of positive goods, goods worth at least three, and fours, and let $y=\ceil{Y/k}$, $a=\ceil{A/k}$,
The selected value pattern is precisely
\[
 4^a,\quad 3^{\,y-a},\quad 2^{\,q-y},\quad y\le q-1,
 \quad t=2q+y+a-1.
\]
There are no ones: a selected two could be replaced by any unselected one. If a three is selected, all twos are selected. If a four is selected, all threes are selected.

Suppose both fours and threes are selected. All threes and twos in the instance are then selected, giving us
\[
 P=A+q-a\le ka+q-a.
\]
But $S_q$ ends in a two, so $P\ge k(q-1)+1$. These inequalities force $q\le a+1$, whereas the presence of both a three and a two forces $q\ge a+2$, a contradiction.

Suppose threes but no fours are selected. Then $A=0$, all $q-y$ twos are selected, and
\[
 k(q-1)+1\le P=Y+q-y\le ky+q-y.
\]
It follows that $q=y+1$, that there is exactly one two, and that $Y=ky$. The target is $t=3y+1$. Each target bundle without the unique two needs at least $y+1$ threes, while its bundle needs at least $y$. This requires at least $ky+k-1$ threes, contradicting $Y=ky$.

In each remaining case there are no threes in the instance. Indeed, if a four is selected, all threes would be selected; if no four or three is selected, $\ceil{Y/k}=0$. All positive values are therefore two or four. Rounding $t$ up to the next even integer $\tau=t+1$ changes no feasibility or acceptance statement and gives $v(B)=\tau$. Lemma~\ref{lem:low-removal} at $\tau$ completes the proof.

This proves property~\ref{item:one} for all three lower families.
\end{proof}

\subsection{Proof of Lemma~\ref{lem:lower-properties}: Property~\ref{item:two}}\label{app:two}
This section proves property~\ref{item:two} for the lower valuations. We continue with the integral values and targets used in Appendix~\ref{app:one}.

For each additive weight below, its value on a bundle is the sum of its singleton weights. We analyze $B^*=S_{q_v}\cup\{u_2\}$ for $k\ge3$. Write $q=q_v$ and $z=v(u_2)$.

\subsubsection{A bound on the removed weight}
The following bound preserves most partition inequalities at once. The proof of property~\ref{item:two} then checks the remaining small targets for each family.

\begin{lemma}\label{lem:two-weight}
Let $w$ be a nonnegative nondecreasing weight on singleton values, with $w(0)=0$, maximum used weight $M_w$, and total weight $W$. For a bundle $B$, write $w(B)=\sum_{g\in B}w(v(g))$. Then $w(B^*)\le\frac Wk+\left(2-\frac1k\right)M_w$.
If the original target criterion is $W\ge k\tau$ and $\tau\ge2M_w$, deleting $B^*$ preserves $W'\ge(k-2)\tau$.
\end{lemma}
\begin{proof}
Express $w$ as a nonnegative linear combination of indicators of nested value prefixes. If such a prefix has size $X$, the full spaced sequence takes $\ceil{X/k}\le X/k+1-1/k$ of its goods. The coefficients sum to $M_w$. Hence the full spaced weight is at most $W/k+(1-1/k)M_w$. The extra good has weight at most $M_w$. Therefore
\[
 W-w(B^*)\ge(k-1)\tau-\left(2-\frac1k\right)M_w
 \ge(k-2)\tau. \qedhere
\]
\end{proof}

\subsubsection{Proof of property~\ref{item:two}}
\begin{proof}[Proof of Lemma~\ref{lem:lower-properties}, property~\ref{item:two}]
If $q=1$, at most two goods are removed, so at most two original partition parts are destroyed. Hence assume $q\ge2$, which gives $t>v(u_1)\ge z$.

We first bound the total value removed. If $q\ge3$ and the last spaced value is $w$, then
\[
 v(S_{q-1})\ge v(u_1)+w\ge z+w,
\]
so integrality gives
\begin{equation}\label{eq:two-total-q3}
 v(B^*)\le t-1+w+z\le2t-2.
\end{equation}
If the last spaced value is one, $v(S_q)=t$, and $v(B^*)<2t$ as well. Thus only a two-good spaced target can require special treatment of total value.

\subsubsection*{Values $\{0,1,2,3,4\}$}
Use cumulative counts $(P,Z,Y,A)$ for values at least $1,2,3,4$, and put $V=P+Z+Y+A$. The exact target-partition conditions are the applicable inequalities in Lemma~\ref{lem:partition-five}. We prove that each survives deleting $B^*$. Write $h=v(u_1)$, so $1\le h\le4$ and $t\ge h+1$.

\paragraph{Total value.}
Equation~\eqref{eq:two-total-q3} handles $q\ge3$. For $q=2$, put $w=v(u_{k+1})$, so
\[
 h\ge z\ge w,\quad h<t\le h+w,\quad v(B^*)=h+z+w.
\]
For $t\ge6$, this is at most $12\le2t$; for $t\le3$, it is at most $3(t-1)\le2t$. Thus an excess over $2t$ has one of exactly three forms: $t=4$ and $(h,z,w)=(3,3,3)$; or $t=5$ and $(h,z,w)=(4,4,3)$ or $(4,4,4)$.

At $t=4$ with three selected threes, $A=0$, the modulus-three row gives $P\ge2k$, and $Z,Y\ge k+1$. Hence $V\ge4k+2$, while the removed value is nine. At $t=5$ with selected values $(4,4,3)$, the row $P+Z+A\ge4k$ and the bound $Y\ge k+1$ give $V\ge5k+1$; the removed value is eleven. At $t=5$ with three selected fours, $P\ge2k$ and $Z,Y,A\ge k+1$, giving $V\ge5k+3$; the removed value is twelve. In each case the residual total is at least $(k-2)t$.

\paragraph{The odd-target row $P+Y$.}
The maximum used singleton weight is $\ceil{h/2}$. Since $t>h$, Lemma~\ref{lem:two-weight} applies except possibly at $t=5$ with $h\in\{3,4\}$. If $q=2$, the removed weight is at most six, exactly twice the target weight. Otherwise $h=3$ and $S_q$ has values $(3,1,1)$, whose weight is four; adding the extra good costs at most two. The row therefore survives.

\paragraph{The row $P+A$.}
If $h\le3$, its maximum used singleton weight is one and its first possible target with $q\ge2$ is $t=4$, with target weight two. Lemma~\ref{lem:two-weight} applies. If $h=4$, it applies for all applicable targets except $t=7$, whose target weight is three.

At $t=7$, $q=2$ removes weight at most six. If $q\ge3$, the first spaced good is four and no later spaced good is four. Moreover, $q\le4$, so the removed weight is at most $q+3$. It can exceed six only when $q=4$, $z=4$, and $S_q$ has values $(4,1,1,1)$. Here $Z,Y\le k$, because the second spaced good is one. Consequently
\[
 P+A=V-Z-Y\ge7k-2k=5k.
\]
Removing weight seven leaves at least $3(k-2)$, as required.

\paragraph{The row $P+Z+A$.}
The maximum used singleton weight is respectively $1,2,2,3$ at maximum values $h=1,2,3,4$. Together with $t>h$, Lemma~\ref{lem:two-weight} handles every applicable target except $t=5,h=4$. In that case $q=2$ and the removed weight can exceed eight only if all three selected goods are fours. As above, $P\ge2k$ and $Z,A\ge k+1$, giving
\[
 P+Z+A\ge4k+2.
\]
The removed weight is nine, so the residual is at least $4(k-2)$.

\paragraph{The three modulus-four rows.}
For the $P$ row, the first possible applicable target is $t=5$, with target weight two and maximum singleton weight one, so Lemma~\ref{lem:two-weight} applies. For $P+Z$, either $h=1$ and the target weight is at least two, or $h\ge2$ and the first possible applicable target is six, with target weight four. The lemma again applies. Finally, for $P+Z+Y$, every applicable target at least seven has target weight at least six and maximum singleton weight at most three. The only smaller possibility is $t=3$, when $h\le2$ and this row is identical to the total-value row already proved.

Every partition inequality is preserved, proving property~\ref{item:two} for this family.

\subsubsection*{Values $\{0,1,d,d+1\}$}
Let $a,b,c$ count values $d+1,d,1$, and set $x=a+b$. It suffices to treat $d\ge3$, since $d=2$ is a subset of $\{0,1,2,3,4\}$. Let $t=d\lambda+\rho=(d+1)\eta+\sigma$, $0\le\rho<d$, and $0\le\sigma\le d$.
The exact criteria are
\begin{equation}\label{eq:two-adj-criteria}
 V=dx+a+c\ge kt,\quad
 W_1=\rho x+a+c\ge k\rho(\lambda+1),\quad
 W_2=\sigma x+c\ge k\sigma(\eta+1).
\end{equation}
Rows with zero residue are automatic. If no large goods remain, the valuation has values in $\{0,1\}$ and is already covered; assume a large good is present.

\paragraph{Total value.}
By \eqref{eq:two-total-q3}, only $q=2$ with two large spaced goods needs attention. Then $x\ge k+1$, and $B^*$ consists of three large goods. Let $a_*$ be the number of its $(d+1)$-goods. If $t\ge2d$, then $v(B^*)\le3d+3\le4d\le2t$.

Otherwise $t=d+\rho<2d$, and the $W_1$ criterion gives us
\[
 V-kt\ge(d-\rho)(x-k)\ge d-\rho.
\]
The removed excess above $2t$ is $d+a_*-2\rho$, so this suffices when $a_*\le\rho$. If $a>0$, $q\ge2$ forces $\rho\ge2$; if $a=0$, then $a_*=0$. The only exception is therefore $\rho=2,a_*=3$. Here $a,x\ge k+1$, while the $W_2$ criterion at $t=d+2$ says $x+c\ge2k$. Consequently
\[
 V=(d-1)x+a+(x+c)\ge k(d+2)+d,
\]
which exceeds the required slack $d-1$. Total value is preserved.

\paragraph{The $W_2$ criterion.}
For $\sigma>0$, $t>v(u_1)\ge d$ implies $\eta\ge1$, except the zero-residue case $t=d+1$. Thus its target weight is at least twice its maximum singleton weight $\sigma$, and Lemma~\ref{lem:two-weight} applies.

\paragraph{The $W_1$ criterion.}
If no $(d+1)$-good remains, its maximum used singleton weight is $\rho$, and $\lambda\ge1$ makes Lemma~\ref{lem:two-weight} applicable. If $(d+1)$-goods remain, the maximum is $\rho+1$. The uniform bound covers $\lambda\ge3$ and $\lambda=2,\rho\ge2$. Since $t>d+1$, the only remaining cases are
\[
 \lambda=1,\quad 2\le\rho<d,
 \quad\text{or}\quad
 \lambda=2,\quad\rho=1.
\]
For $\lambda=1$, if $x\le k$ then the first spaced target is one $(d+1)$-good and $\rho-1$ ones. Its $W_1$-weight equals the target $2\rho$; adding one singleton costs at most $\rho+1\le2\rho$. If $x\ge k+1$, then $q=2$ and the removed weight is $3\rho+a_*$. It is at most $4\rho$ except when $\rho=2,a_*=3$. In that exception,
\[
 W_1=(x+c)+x+a\ge2k+(k+1)+(k+1)=4k+2,
\]
using the same $W_2$ row as above. The removed excess is only one.

For $\lambda=2,\rho=1$, the target $W_1$-weight is three. If $q\le2$, the removed weight is at most six. Otherwise $x\le k$: if $x\ge k+1$, the first two spaced large goods, including a $(d+1)$-good, already reach $2d+1$. Hence $q=d+1$, and
\[
 W_1=V-(d-1)x\ge k(d+2),\quad W_1(B^*)\le d+4.
\]
The residual weight is at least $3(k-2)$, since
\[
 k(d+2)-(d+4)-3(k-2)=k(d-1)-d+2\ge0.
\]
All criteria in \eqref{eq:two-adj-criteria} survive. This proves property~\ref{item:two} for this family.

\subsubsection*{Values $\{0,1,2,2e\}$}
Let $H=2e\ge4$, and let $a,b,c$ count values $H,2,1$. If no $H$-goods remain, the valuation is a subset of $\{0,1,2,3,4\}$. Otherwise $q\ge2$ implies $t>H$. Write
\[
 t=H\lambda+r,\quad 0\le r<H,\quad R=\ceil{r/2},\quad s=\ceil{t/2}=e\lambda+R.
\]
The exact criteria are
\begin{align}\label{eq:two-even-criteria}
 V&=Ha+2b+c\ge kt, & W_1&=ra+2b+c\ge kr(\lambda+1),\\
 W_2&=ea+b+c\ge ks, & W_3&=Ra+b+c\ge kR(\lambda+1).\nonumber
\end{align}
Rows with $r=0$ involving $W_1,W_3$ are automatic, and $\lambda\ge1$ throughout. Put $\alpha=\ceil{a/k}$ and $\zeta=\ceil{(a+b)/k}$.

\paragraph{Total value.}
Only $q=2$ requires attention. If the second spaced good is worth at most two, then $v(B^*)\le2H+2\le2t$. Otherwise $a\ge k+1$ and $B^*$ consists of three $H$-goods. The case $t\ge3H/2$ is immediate. In the remaining case $t=H+r$ with $r>0$, the $W_1$ criterion gives
\[
 V-kt\ge(H-r)(a-k)\ge H-r.
\]
The removed excess is $3H-2t=H-2r$, no larger than that slack.

\paragraph{The $W_3$ criterion.}
For $r>0$, its maximum singleton weight is $R$, and its target weight $R(\lambda+1)$ is at least $2R$. Lemma~\ref{lem:two-weight} applies.

\paragraph{The $W_2$ criterion.}
The uniform bound applies for $s\ge2e$, that is, $t\ge2H-1$. Otherwise $t=H+r$, $1\le r\le H-2$, and $s=e+R$.

If $\alpha\ge2$, then $q=2$ and $W_2(B^*)=3e$. From the $W_3$ criterion,
\[
 W_2\ge2kR+(e-R)a\ge k(e+R)+(e-R).
\]
This slack exceeds the removed excess $3e-2(e+R)=e-2R$.

If $\alpha=1$ and twos suffice for the spaced target, then $q=1+R$ and $W_2(S_q)=s$. Adding one good costs at most $e<s$, so the row survives. If ones are needed, then
\[
 2(\zeta-1)<r,\quad q=r-\zeta+2,\quad
 L:=W_2(S_q)=e+r-\zeta+1\ge s.
\]
Using $a\le k$, $a+b\le k\zeta$ and the total-value criterion gives
\[
 W_2=V-(e-1)a-(a+b)\ge kL.
\]
Thus $W_2-W_2(B^*)\ge(k-1)L-e\ge(k-2)s$.

\paragraph{The $W_1$ criterion.}
For $r\ge2$, the singleton weights are nondecreasing with value and have maximum $r$, so Lemma~\ref{lem:two-weight} applies. For $r=1$, the weights at values $1,2,H$ are $1,2,1$, respectively, so we use a direct count instead. The weight of the full spaced sequence is
\[
 \ceil{(a+b+c)/k}+\ceil{(a+b)/k}-\ceil{a/k}
 \le \frac{W_1}{k}+2-\frac2k.
\]
All weights are nonnegative, and the additional good has weight at most two. Thus
\[
 W_1(B^*)\le\frac{W_1}{k}+4-\frac2k.
\]
When $\lambda\ge3$, the original inequality $W_1\ge k(\lambda+1)$ gives
\[
 W_1-W_1(B^*)\ge(k-1)(\lambda+1)-4+\frac2k
 \ge(k-2)(\lambda+1).
\]
Only $t=\lambda H+1$ with $\lambda\in\{1,2\}$ remains. Its target weight is $\tau=\lambda+1$.

If $\alpha\ge\lambda+1$, then $q=\lambda+1$ and all $q+1$ goods of $B^*$ are high, so their weight is $\lambda+2\le2\tau$. If $\alpha\le\lambda-1$, let $L=t-(H-1)\alpha$.
The total-value criterion gives $W_1\ge kL$. The spaced target ends in a two or a one, so $W_1(S_q)\le L+1$ and $W_1(B^*)\le L+3$. Since $L-\tau=(H-1)(\lambda-\alpha)\ge3$,
the residual weight is at least $(k-2)\tau$.

Finally let $\alpha=\lambda$. The positive-good count is at least $k\lambda+1$, since every good is worth at most $H$ and $V\ge k(\lambda H+1)$. Hence $q=\lambda+1$. If $\lambda=2$, the second good is high and $W_1(B^*)\le5\le6=2\tau$. If $\lambda=1$ and $a\ge2$, the second good is again high and $W_1(B^*)\le4=2\tau$. In the remaining case $\lambda=a=1$,
\[
 W_1=V-(H-1)\ge2k+(k-1)(H-1),\quad W_1(B^*)\le5.
\]
Its original slack is at least six, while the removed excess is at most one.

Thus all four inequalities are preserved, proving property~\ref{item:two} for this family and completing the proof of Lemma~\ref{lem:lower-properties}.
\end{proof}

\section{Proof of Theorem~\ref{thm:extension}}\label{app:extension}
We give the complete proof of Theorem~\ref{thm:extension}. We first derive the consequence of property~\ref{item:two} used to account for additional lower-valued removals.

\subsection{Additional Removals}

\begin{lemma}\label{lem:additional-removals}
Fix an ordered valuation $v$, an integer $k\ge3$, a positive feasible $k$-target $t$, and its first acceptable spaced size $q_v$. Suppose removing $S_{q_v}\cup\{u_2\}$ leaves a $(k-2)$-partition at $t$. For every $1\le d\le k-2$, removing any set $B$ satisfying
\[
 |B|\le q_v+d,\quad
 |B\cap\{u_1,\dots,u_X\}|\le\ceil{X/k}+d
 \quad\text{for every prefix length }X
\]
leaves a $(k-d-1)$-partition at $t$.
\end{lemma}
\begin{proof}
When $d=1$, the assumed rank and cardinality bounds imply that the increasing ranks of $B$ are componentwise at least the corresponding ranks in $S_{q_v}\cup\{u_2\}$. Indeed, the first two ranks are at least $1,2$, and the $j$th rank for $j\ge3$ is at least $(j-2)k+1$. Consequently, after deleting $B$, the $j$-th largest remaining value is at least the $j$-th largest value remaining after deleting $S_{q_v}\cup\{u_2\}$, for every corresponding position. There may also be extra goods. A target partition therefore transfers to the complement of $B$.

For $d>1$, remove the first $d-1$ members of $B$ in rank order and call them $E$. If $|B|<d-1$, deleting $B$ directly destroys at most $d-1$ parts of the original partition. Otherwise, $B\setminus E$ satisfies the $d=1$ bounds: every prefix containing more than $d-1$ members of $B$ loses exactly $d-1$ of them. Its complement has a $(k-2)$-partition. Deleting the remaining $d-1$ goods in $E$ destroys at most $d-1$ of those parts, leaving at least $k-d-1$ target bundles. Any unused goods or excess bundles can be merged into the retained bundles.
\end{proof}

\subsection{Proof of the Extension Theorem}

\begin{proof}[Proof of Theorem~\ref{thm:extension}]
Let $U$ be the lower goods. An empty added class changes nothing, so assume it contains at least one good. Fix a target $t>0$ admitting a $k$-partition, and write the upper count as $a=k\ell+r$. For a removed set $B$, let $b$ be its number of upper goods and $c$ its number of lower goods. All lower-prefix bounds below refer to ranks within $U$. We prove property~\ref{item:one} and, for every $1\le d\le k-2$, the conclusion of Lemma~\ref{lem:additional-removals}. The latter implies property~\ref{item:two}, because $S_{q_v}\cup\{u_2\}$ meets its bounds with $d=1$. For the lower valuation, that conclusion follows from property~\ref{item:two}.

\paragraph{Initial cases.}
Suppose $t\le\ell H$ and put $q_0=\ceil{t/H}\le\ell$. The first acceptable spaced size is $q_0$. Property~\ref{item:one} follows from cardinality alone: deleting at most $q_0$ goods leaves at least $(k-1)q_0$ upper goods.

For the additional-removal conclusion, deleting at most $q_0+d$ goods leaves at least $(k-d-1)q_0$ upper goods, because $k\ell+r-(q_0+d)-(k-d-1)q_0
 =k(\ell-q_0)+r+d(q_0-1)\ge0$.
These upper goods alone meet all required targets.

Now suppose $t>\ell H$. By Lemma~\ref{lem:recursion}, let $s=t-\ell H$ and $0<s\le\MMS_{k-r}(v,U)\le H$.

If $\ell=0$, then $r>0$ and $q_v=1$. Property~\ref{item:one} deletes at most one good, while the additional-removal bounds permit at most $d+1$ goods; discarding the original partition parts containing those goods proves both properties.

Henceforth assume $\ell\ge1$. If $r>0$, then $q_v=\ell+1$. If $r=0$, the lower goods have a $k$-partition at $s$, and $q_v=\ell+q_U$, where $q_U$ is the lower valuation's first acceptable $k$-spaced size at $s$. We now prove the two properties separately.

\paragraph{Property~\ref{item:one}.}
Let $B$ satisfy the premises of property~\ref{item:one}. Suppose first that $r>0$, so $b+c\le\ell+1$. After giving $\ell$ upper goods to each of $k-1$ residual bundles, the number of spare upper goods is $z=a-b-(k-1)\ell=\ell+r-b\ge0$.
If $z\ge k-1$, every residual bundle is covered. Otherwise, at most $k-1-z=(k-r)-(\ell+1-b)$ residual bundles need lower goods at target $s$. Deleting $c\le\ell+1-b$ lower goods destroys at most that many parts of a $(k-r)$-partition of $U$ at $s$, leaving enough.

Now let $r=0$. The upper-prefix spacing bound gives $b\le\ell$. Put $e=\ell-b$. The lower removed set satisfies
\begin{equation}\label{eq:extension-lower-counts}
 c\le q_U+e,\quad
 |B\cap\text{lower prefix }X|\le\ceil{X/k}+e.
\end{equation}
There are exactly $e$ spare upper goods after allocating $\ell$ to each residual bundle. If $e\ge k-1$, these suffice. If $1\le e\le k-2$, apply Lemma~\ref{lem:additional-removals} to the lower valuation and \eqref{eq:extension-lower-counts}. It leaves $k-e-1$ lower target bundles, and the $e$ spare upper goods cover the rest.

It remains to treat $e=0$. If $c=0$, there is nothing to remove below the upper class. Otherwise $c\le q_U$, and the removed lower set is $k$-spaced. Keep the chosen $\ell$ upper goods fixed. Any lexicographically later, acceptable lower $c$-set could replace the lower part of $B$ and would give a lexicographically later acceptable member of the original $\mathcal F_{k,q}$: its lower ranks, shifted by $a=k\ell$, satisfy the required global spacing. This contradicts the comparison assumed for $B$. Property~\ref{item:one} for the lower valuation therefore leaves a $(k-1)$-partition at $s$.

\paragraph{Property~\ref{item:two}.}
We prove the stronger additional-removal conclusion. Fix $1\le d\le k-2$, let $B$ satisfy the bounds of Lemma~\ref{lem:additional-removals}, and put $h=k-d-1$.

Suppose first that $r>0$. The spare upper count is
\[
 z=a-b-h\ell=(d+1)\ell+r-b\ge d(\ell-1)+r-1\ge0,
\]
since $b+c\le\ell+1+d$. If $z\ge h$, upper goods suffice. Otherwise, the number of lower target bundles required is
\[
 h-z=(k-r)-D,\quad D=(d+1)(\ell+1)-b.
\]
But $c\le\ell+1+d-b\le D$. Delete the at most $c$ affected parts of a lower $(k-r)$-partition. Spare upper goods cover the other residual bundles.

Now let $r=0$. The upper-prefix bound gives $b\le\ell+d$. The spare upper count is
\[
 z=(d+1)\ell-b\ge d(\ell-1)\ge0.
\]
If $z\ge h$, upper goods suffice. Otherwise set $d'=d+z=(d+1)\ell+d-b$ for all $1\le d'\le k-2$.
The lower removed set obeys
\[
 c\le q_U+\ell+d-b\le q_U+d',\quad
 |B\cap\text{lower prefix }X|
 \le\ceil{X/k}+\ell+d-b\le\ceil{X/k}+d'.
\]
Applying Lemma~\ref{lem:additional-removals} to the lower valuation leaves $k-d'-1=h-z$ lower target bundles. The $z$ spare upper goods cover the remaining residual bundles.

In both arguments for $t>\ell H$ and $\ell\ge1$, every constructed lower bundle is combined with $\ell$ upper goods, and every upper-only bundle has at least $\ell+1$ upper goods; their values are at least $t$. Extra goods can be assigned arbitrarily. The condition in Definition~\ref{def:domain} remains true after deletions, so the argument iterates over any finite number of classes.
\end{proof}

\section{Running-Time Bounds in Theorem~\ref{thm:main}}\label{app:implementation}
We justify the running time in Theorem~\ref{thm:main}, including recognition of the valuation domain and exact computation of all targets. The counts and inequalities in this section are evaluated separately for each agent; no common value classes are required.

\subsection{Finding a Lower Part}
After sorting an agent's values, group equal positive values and compute their multiplicities and the total value below each class. The lower part can contain at most four distinct positive values. Hence there are at most five possible choices for it: the first zero, one, two, three, or four positive values in increasing order. For each choice, check the condition in Definition~\ref{def:domain} at every remaining class. This takes linear time in the number of goods because the number of choices is constant.

Membership of the proposed lower values in one of the three scaled families can also be checked with constant time. An empty set or a single positive value is immediate. For the family $\{0,1,2,3,4\}$, let $x$ be any positive lower value and try the four scales $x/j$, $j\in\{1,2,3,4\}$; one is the correct scale whenever such a representation exists.

For $\{0,1,d,d+1\}$ with three distinct positive lower values, the smallest is the scale. Check that the normalized values are $1,d,d+1$ for an integer $d\ge2$. With two positive values $x<y$, it suffices to try scales $x$ and $y-x$. The first gives a representation exactly when $y/x$ is an integer at least two. The second gives a representation as $d,d+1$ exactly when $x/(y-x)$ is an integer at least two. These alternatives include every two-positive-value subset of the adjacent family.

For $\{0,1,2,2e\}$ with three distinct positive values, the smallest is again the scale; check that the normalized middle value is two and the normalized largest value is an even integer at least four. Every two-positive-value subset of this family has an integer ratio and was already considered above. No family other than $\{0,1,2,3,4\}$ has four distinct positive lower values. These tests are therefore complete. They do not search over the possible numerical sizes of $d$ or $e$.

Deleting goods preserves any accepted description, so it need not be found again during allocation. Sorting all agents' values and recognizing the domain takes $O(nm\log(2+m))$ time.

\subsection{Computing Exact MMS Values}
For an accepted description, start with bundle count $k=n$ and total zero, and process the higher value classes in decreasing order. At a class with value $H$ and multiplicity $a$, Lemma~\ref{lem:recursion} adds $\floor{a/k}H$ to the total and replaces $k$ by $k-(a\bmod k)$. The number of bundles remains positive. After the last higher class, evaluate the relevant formula from Appendix~\ref{app:partitions} at the resulting $k$, multiply by the private scale, and add it to the total. An empty lower part contributes zero.

Every class takes a constant number of operations. The formulas in Corollaries~\ref{cor:mms-adj}, \ref{cor:mms-even}, and~\ref{cor:mms-five} also take a constant number once the lower counts are known. Thus all agents' exact MMS values require $O(nm)$ operations after sorting. If targets below MMS are supplied instead, the allocation step uses those fixed targets.

\subsection{Finding the Rightmost Acceptable Bundle}
Retain each agent's original sorted array of singleton values, indexed by the common auxiliary ranks. Maintain the set of unallocated ranks, and, for each agent, a binary indexed tree \citep{fenwick1994} storing her value at each unallocated rank and zero at each allocated rank. Such a tree uses linear space and permits prefix-sum queries and changes of one array entry in logarithmic time. Differences of two prefix sums give the value of an interval.

At the start of an iteration, form the list mapping current ranks $1,\dots,p$, where $p=|U|$, to original auxiliary ranks. This costs $O(m)$ per iteration and $O(nm)$ overall. With $k$ agents remaining, scan the spaced positions until some target is met. Finding $q$ and all agents' partial sums along $S_q$ takes $O(kq)$ operations. Store also the suffix sums along these same $q$ spaced positions; this takes another $O(kq)$ operations.

Suppose ranks $r_1,\dots,r_{j-1}$ are fixed and the proposed next rank is $r$. The earliest feasible completion is given by \eqref{eq:completion}. It is acceptable whenever any completion of this proposed prefix is acceptable, since it is coordinate by coordinate at least as valuable to every agent. To evaluate it, observe that its first
\[
 c=\min\left\{q-j+1,\ 1+\floor{\frac{r-((j-1)k+1)}{k-1}}\right\}
\]
unfixed ranks form the consecutive interval $r,\dots,r+c-1$. Its remaining ranks are the corresponding suffix of $S_q$. The interval maps to an interval of original ranks; deleted entries contribute zero in every tree. For each agent, the proposed bundle value is therefore her already fixed prefix value, one interval sum, and one stored spaced-suffix sum. This takes $O(\log(2+m))$ time.

Candidate acceptability is nonincreasing in $r$. Binary search over
\[
 \max\{(j-1)k+1,r_{j-1}+1\}\le r\le p-q+j
\]
finds the largest acceptable rank. For $j=1$, the lower bound is one. The interval of candidates is nonempty because the previously fixed prefix has an acceptable completion. Once $r_j$ is fixed, add its singleton value to each agent's fixed-prefix sum and continue with $j+1$. Every search includes the current accepting possibilities of all remaining agents, not a recipient selected in advance. Thus this procedure returns exactly the bundle in Algorithm~\ref{alg:exact} and leaves its choice of recipient unrestricted.

The cost per coordinate is $O(k\log^2(2+m))$. Each allocated bundle contains $q$ actual goods, so $\sum q\le m$ and $\sum kq\le nm$ over all iterations. The coordinate searches therefore cost $O(nm\log^2(2+m))$ operations in total. Updating all trees after removals costs $O(nm\log(2+m))$, and the spaced scans and current-rank lists cost $O(nm)$.

For the final transfer, keep each agent's original preference list and a pointer to its first unallocated good. Whenever she receives an auxiliary rank, advance her pointer past allocated goods and assign its current good. Each pointer advances at most $m$ times, so this step costs $O(nm)$. The sorted arrays, trees, and lists use $O(nm)$ space. The all-zero-target and single-remaining-agent cases require no bundle search and satisfy the same bounds.

\end{document}